\documentclass{article}

\usepackage{mathptmx}

\usepackage{amsthm}
\usepackage{amsmath}
\usepackage{mathbbol}
\usepackage{amssymb}
\usepackage{graphicx}
\usepackage{url}
\usepackage{epsfig}

\usepackage{algorithm}
\usepackage{algpseudocode}

\newtheorem{theorem}{Theorem}
\newtheorem{lemma}[theorem]{Lemma}

\newtheorem{corollary}[theorem]{Corollary}
\newtheorem{definition}[theorem]{Definition}

\numberwithin{equation}{section}
\numberwithin{figure}{section}

\newcommand{\N}{\mathbb{N}}
\newcommand{\R}{\mathbb{R}}
\newcommand{\C}{\mathbb{C}}
\newcommand{\Sig}{\Sigma_f}
\newcommand{\Otilde}{\tilde{O}}

\newcommand{\aqirarrow}{\stackrel{\textsc{Aqir}}{\rightarrow}}

\newcommand{\Z}{\mathbb{Z}}
\renewcommand{\a}{\alpha}

\newcommand{\eps}{\epsilon}

\newcommand{\sgn}{\mathrm{sign}}
\newcommand{\res}{\mathrm{res}}

\newcommand{\IBox}{\mathfrak{B}}

\newcommand{\down}{\mathrm{down}}
\newcommand{\up}{\mathrm{up}}

\newcommand{\RB}{\Gamma}
\newcommand{\ND}{\textsc{New}\textsc{Dsc}}

\begin{document}

\title{Root Refinement for Real Polynomials.}

\author{Michael Kerber\thanks{Stanford University, Stanford, USA and Max-Planck-Center for Visual Computing and Communication, Saarbr\"ucken, Germany ({\tt mkerber@mpi-inf.mpg.de}).}
        \and Michael Sagraloff\thanks{Max-Planck Institut f\"ur Informatik, Saarbr\"ucken, Germany ({\tt msagralo@mpi-inf.mpg.de}).}}

\maketitle

\begin{abstract}
We consider the problem of approximating all real roots of a square-free polynomial $f$. 
Given isolating intervals, our algorithm refines each of them 
to a width of $2^{-L}$ or less, that is, each of the roots is approximated to $L$ bits after the binary point.
Our method provides a certified answer for arbitrary real polynomials, only considering finite
approximations of the polynomial coefficients and choosing a suitable working precision adaptively.
In this way, we get a correct algorithm that is simple to implement and practically efficient.
Our algorithm uses the quadratic interval refinement method; we adapt that method to be able to cope with
inaccuracies when evaluating $f$, without sacrificing its quadratic convergence behavior.
We prove a bound on the bit complexity of our algorithm in terms of the degree of the polynomial, the size
and the separation of the roots, that is, parameters exclusively related to the geometric location of the roots. 
Our bound is near optimal and significantly improves previous work on integer polynomials. Furthermore, it essentially matches the best known theoretical 
bounds on root approximation which are obtained by
very sophisticated algorithms. 
We also investigate the practical behavior of the algorithm and demonstrate how closely the practical
performance matches our asymptotic bounds.
\end{abstract}

\section{Introduction}
\label{sec:intro}

The problem of computing the real roots of a polynomial in one variable
is one of the best studied problems in mathematics.
If one asks for a \emph{certified}
method that finds all roots, 
it is common to write the solutions as a set of disjoint
\emph{isolating} intervals, each containing exactly one root; for that reason,
the term \emph{real root isolation} is common in the literature.
Simple, though efficient methods for this problem have been presented,
for instance, based on Descartes' rule of signs~\cite{ca-prriudrs-76}, 
or on Sturm's theorem~\cite{dusharmayap-dmbound2007}.
Recently, the focus of research shifted to polynomials with real
coefficients which are approximated during the algorithm.
It is worth remarking that this approach does not just generalize the integer
case but has also leads to practical~\cite{ek+-descartes,RS} and theoretical~\cite{sag-complexity} 
improvements of it.

We consider the related \emph{real root refinement problem}: assuming that isolating intervals
of a polynomial are known, \emph{refine} them to a width of $2^{-L}$ or less,
where $L\in\mathbb{N}$ is an additional input parameter.
Clearly, the combination of root isolation and root refinement,
also called \emph{strong root isolation}, yields 
a certified approximation of all roots of the polynomial
to an absolute precision of $2^{-L}$ or, in other words, to $L$ bits 
after the binary point in binary representation.

We introduce an algorithm, called $\textsc{Aqir}$, to solve the root refinement problem for arbitrary square-free
polynomials with real coefficients. Most of the related approaches are formulated
in the REAL-RAM model where exact operations on real numbers are assumed to be
available at unit cost.
In contrast, our approach works only with approximations of the input and exclusively performs
approximate but certified arithmetic. Here, we assume the existence of an oracle which, for an arbitrary positive integer $\rho$, provides approximations of the coefficients of the input
polynomial to an error of less than $2^{-\rho}$. In the analysis of our algorithm, we also quantify the size of $\rho$ in the worst case. 
The refinement uses the quadratic interval refinement method~\cite{abbott-quadratic} (QIR for short)
which is a quadratically converging hybrid of the bisection and the secant method.
We adapt the method to work with an increasing working precision
and use interval arithmetic to validate the correctness of the outcome.
In this way,
we obtain an algorithm that always returns a correct root approximation,
is simple to implement on an actual computer 
(given that arbitrary approximations of the coefficients are accessible),
and is adaptive in the sense that it might succeed with a much lower
working precision than predicted by the worst-case bound.

We provide a bound on the bit complexity of our algorithm. To state it properly, we first define several
magnitudes depending on the polynomial which remain fixed throughout the paper.
Let
\begin{align}
f(x):=\sum_{i=0}^{d}a_{i}x^{i}\in\R[x]\label{polyf}
\end{align}
be a square-free polynomial of degree $d\geq 2$ with $|a_d|\ge 1$ and $\tau:=\left\lceil \log(\max_i |a_i|)\right\rceil\ge 1$ (throughout the paper, $\log$
means the logarithm with base $2$). We denote the (complex) roots of $f$ by $z_1,\ldots,z_d$, and, w.l.o.g., we can assume that the roots are numbered such that the first $m$ roots $z_1,\ldots,z_m$ are exactly the real roots of $f$.
For each $z_{i}$, $\sigma_i=\sigma(z_{i},f):=\min_{j\neq 
i}|z_{i}-z_{j}|$ denotes the \emph{separation of $z_{i}$}, 
$\Sigma_f:=\sum_{i=1}^{d}\log \sigma_i^{-1}$ and $\Gamma_f:=\max(1,\log(\max_i |z_i|))$ the \emph{logarithmic root bound} of $f$. An interval $I=(a,b)$ is called 
\emph{isolating} for a root $z_i$ if $I$ contains $z_i$ and no other 
root of $F$. We set $\operatorname{mid}(I)=\frac{a+b}{2}$ for the \emph{center} and 
$w(I):=b-a$ for the \emph{width} of $I$.\\

\textbf{\emph{Main Result.}} \emph{Given initial isolating intervals for the roots of $f$, 
our algorithm refines all intervals to the width $2^{-L}$ using 
\begin{align}\label{ref:mainresult}
\tilde{O}(d(d\Gamma_f+\Sigma_f)^2+dL)
\end{align}
bit operations, where $\Otilde$ means that we ignore logarithmic factors.
To do so, our algorithm requires the coefficients of $f$ \emph{at a precision} of at most
$$\tilde{O}(d\Gamma_f+\Sigma_f+L)$$
bits after the binary point.}\\\\
We remark that, if $L$ dominates all other input parameters, the bound in (\ref{ref:mainresult}) is optimal up to logarithmic factors because reading the output already takes $\Theta(dL)$ bit operations in the presence of $m=\Theta(d)$ real roots.

For the analysis, we divide the sequence of QIR steps in the refinement process
into a \emph{linear sequence} where the method behaves like bisection in the worst case,
and a \emph{quadratic sequence} where the interval is converging quadratically
towards the root, following the approach in~\cite{kerber-complexity}.
We do not require any conditions on the initial intervals except that they are
disjoint and cover all real roots of $F$;
an initial \emph{normalization phase} 
modifies the intervals to guarantee the efficiency of our refinement strategy.

We give two variants of our algorithm; for the first variant which we consider to be more 
practical, we use approximate polynomial evaluation at single points only, whereas, for the second 
(more theoretical) variant, we group up to $n$ evaluations together and use fast approximate 
multipoint evaluation~\cite{ks-fast}. The idea behind the second approach is that we perform polynomial evaluations simultaneously for all intervals at the same cost as for a single classical evaluation.\footnote{Very recent work~\cite{pantsi:ISSAC13} introduces an alternative method for real root refinement which, for the task of refining a single isolating interval, achieves comparable running times as \textsc{Aqir}. In a preliminary version of their conference paper (which has been sent by the authors to M.~Sagraloff in April 2013), the authors claim that using approximate multipoint evaluation also yields an improvement by a factor $n$ for their method. Given the results from this paper, this seems to be correct, however, their version of the paper did not contain a rigorous argument to bound the precision demand for the fast multipoint evaluation. This has been achieved first in~\cite{ks-fast}.} 
This yields the complexity bound in (\ref{ref:mainresult}) which is by a factor $m\le d$ better (if $L$ dominates all other input parameters) than the bound $\tilde{O}(d(d\Gamma_f+\Sigma_f)^2+m\cdot dL)=\tilde{O}(d(d\Gamma_f+\Sigma_f)^2+d^2L)$ as achieved by the variant without multipoint evaluations.   

We remark that, using the root solver from~\cite{sag-complexity}, initial isolating intervals can be obtained with
$\Otilde(d(d\Gamma_f+\Sigma_f)^2)$ bit operations using coefficient approximations of $f$ to
$\tilde{O}(d\Gamma_f+\Sigma_f)$ bits after the binary point. Hence, 
our complexity result from (\ref{ref:mainresult}) also gives a bound on the strong root isolation problem.
 
The case of integer coefficients is often of special interest, and, with respect to the QIR method, the problem has been 
investigated by previous work~\cite{kerber-complexity} for this specific case.
In that work, the complexity of root refinement was bounded by
$\Otilde(d^4\tau^2+m\cdot d^2L)=\Otilde(d^4\tau^2+d^3L)$.
We lower this bound and arrive at a complexity of
\begin{align}
\Otilde(d^3\tau+dL).\label{mainresult:integer}
\end{align}
The improvement stems from a combination of several ideas that we describe separately:
In comparison to the purely exact method from~\cite{kerber-complexity}, we get rid of one factor of $d$ 
because, for $\textsc{Aqir}$, we consider a different approach for evaluating
the sign of $f$ at rational points (the main operation in the refinement procedure) than for the classical QIR method:
for an interval of size $2^{-\ell}$,
the evaluation of $f$ at the endpoints of the interval has a complexity of $\Otilde(d^2(\tau+\ell))$ when 
using exact rational arithmetic because the function values can
consist of up to $d(\tau+\ell)$ bits.
However, we show that we can still compute the sign of the function value
with certified numerical methods
using the substantially smaller working precision of $O(d\tau+\ell)$. We remark that the latter result certainly only applies to points whose distance to a root is not much smaller than $2^{-\ell}$, thus, for \textsc{Aqir}, we modified the QIR method in way such that the latter requirement is assured;
this improvement is described in Sections~\ref{sec:aqir} to~\ref{sec:root_refinement}. The latter modifications yield an algorithm with bit complexity
$\tilde{O}(d^3\tau^2+m\cdot dL)=\tilde{O}(d^3\tau^2+d^2L)$
for all real roots. Another factor of $m$ in the second term is then shaved off by using approximate multipoint evaluation in the algorithm as already mentioned above.

Finally, we mix the ideas from~\cite{sag-complexity} with our approach. 
The $\tau^2$ term in the complexity is due to the fact that, in the worst case, our refinement algorithm performs bisections 
until the isolating interval has reached a certain threshold. We change the algorithm such that it performs a hybrid
of bisections and Newton-like steps initially, and switches to QIR after reaching the threshold. This further reduces the complexity to
$\tilde{O}(d^3\tau+dL).$
We remark that this last optimization is restricted to the case of integer polynomials, whereas the first two improvements
apply to our general setup and lead to the main result (\ref{ref:mainresult}) stated above.

We have implemented the exact version of the QIR algorithm and the approximate variant \textsc{Aqir} that realizes our first improvement step.
We report on experimental results
when applying both versions to two families of random input instances. We focus on the comparison
of both variants when increasing one of the input parameters. We demonstrate that, for increasing degree
of the input polynomial, refining a single root scales quadratically for the exact version
and linearly for the approximate version. Hence, by choosing a smaller working precision,
we get rid of a factor of $d$ both in theory and in practice.\\

\textbf{\emph{Related work.}} 
The problem of accurate root approximation is omnipresent in mathematical applications;
certified methods are of particular importance in the context of computations
with algebraic objects, for instance, when computing the topology
of algebraic curves~\cite{RUR-topology,ekw-fast} or when solving systems of multivariate
equations~\cite{bes-bisolve11}.

The idea of combining bisection with a faster
converging method to find roots of continuous functions
has been first introduced in \emph{Dekker's method}
and elaborated in \emph{Brent's method}; see~\cite{bd-two} for a summary.
However, these approaches assume exact arithmetic for their convergence results.

For polynomial equations, numerous algorithms are available,
for instance, the \emph{Jenkins-Traub algorithm} or \emph{Durand-Kerner iteration};
although they usually approximate the roots very fast in practice~\cite{BF00}, 
general worst-case bounds on their arithmetic complexity 
are not available. In fact, for some variants, even termination cannot be
guaranteed in theory; we refer to the survey~\cite{pan-survey} for extensive references
on these and further methods. 

The theoretical complexity of root approximation has been investigated by Pan~\cite{Pan:alg,pan-optimal}.
Assuming all roots to be in the unit disc, he achieves a bit complexity
of $\Otilde(n^3+n^2L)$ for approximating all roots to an accuracy of $2^{-L}$, which matches
our bound (if $L$ is the dominant input parameter) for the first variant of \textsc{Aqir} which does not use fast multipoint evaluation.
His approach even works for polynomials with multiple roots.
However, as Pan admits in~\cite{pan-survey},
the algorithm is difficult to implement and so is the complexity analysis when taking rounding errors
in intermediate steps into account. In addition, the method is global in a sense that all complex roots are approximated in parallel, hence it does not profit from situations where the number of real roots is small. 
A very recent approach~\cite{mesawa:ISSAC13} for root isolation and refinement uses Pan's factorization algorithm from~\cite{Pan:alg} as a key ingredient. For square-free polynomials, the corresponding algorithm achieves a comparable bit complexity bound for the refinement as our asymptotically fast variant of \textsc{Aqir}. However, this does not turn our present results obsolete: First, our approach considerably differs from all existing global root finding algorithms which combine the splitting circle method~\cite{schonhage:fundamental} with techniques from numerical analysis  (Newton iteration, GraeffeÕs method, discrete
Fourier transforms) and fast algorithms for polynomial and
integer multiplication. Second, our algorithm is adaptive in the sense that its computational complexity is directly related to the number of real roots which is often much smaller than the degree of the polynomial. Third, because of its simpleness and the low algorithmic overhead\footnote{In particular, this holds for the first variant of \textsc{Aqir} which does not use fast multipoint evaluation. We consider the second variant based on fast approximate multipoint evaluation more to be a theoretical proof of concept that the overall approach may yield almost optimal complexity bounds.}, it is well suited for an efficient implementation.

We improve upon the conference version of this paper~\cite{ks-complexity} in several ways:
in our bit complexity result, we remove the dependence on the coefficient size and, thus, relate the hardness of root approximation to parameters that exclusively depend on the geometric location of the roots.
In addition, we redefine the threshold for the interval
width that guarantees quadratic convergence (Definition~\ref{def:C_bound}); 
in this way, we get rid of the magnitude 
$R=\log|\res(f,f')|^{-1}$, which is a pure artifact of the analysis of~\cite{ks-complexity}. 
Moreover, the improvements on the complexity result using multipoint evaluations and hybrid Newton steps, 
as well as the experimental evaluations did not appear in~\cite{ks-complexity}.\\

\textbf{\emph{Outline.}} We summarize the (exact) QIR method
in Section~\ref{sec:eqir}. Our \textsc{Aqir} algorithm that only uses 
approximate coefficients
is described in Section~\ref{sec:aqir}. 
Its precision demand is analyzed in Section~\ref{apxivarithemtic}.
Based on that analysis of a single refinement step, 
the complexity bound 
of root refinement is derived in Section~\ref{sec:root_refinement}.
Some experimental comparison between QIR with exact and approximate
coefficients is presented in Section~\ref{sec:experiments}.
Further asymptotic improvements using multipoint evaluation
and special techniques for integer polynomials
are described in Section~\ref{sec:tricks}.
We end with concluding remarks in Section~\ref{sec:conclusion}.

\section{Review on exact QIR}
\label{sec:eqir}

\begin{algorithm}[t]
\caption{\textsc{Eqir}: Exact Quadratic Interval Refinement}
\label{alg:eqir}
\textsc{Input:} $f\in\R[x]$ square-free, $I=(a,b)$ isolating, $N=2^{2^i}\in\N$
\newline
\textsc{Output:} $(J,N')$ with $J\subseteq I$ isolating for $\xi$ and $N'\in\N$
\begin{algorithmic}[1]
\Procedure {eqir}{$f,I=(a,b),N$}
  \State \textbf{if} $N=2$, \textbf{return} $($\Call{Bisection}{$f,I$},$4)$.
  \State $\omega\gets\frac{b-a}{N}$
  \State \label{eqir:point_location}$m'\gets a+\mathrm{round}(N\frac{f(a)}{f(a)-f(b)}) \omega$
\Comment $m'\approx a+\frac{f(a)}{f(a)-f(b)}(b-a)$
  \State \label{eqir:evaluation_begin}$s\gets \sgn(f(m'))$
  \State \textbf{if} $s=0$, \textbf{return} $([m',m'],\infty)$
  \State \textbf{if} $s=\sgn(f(a))$ \textbf{and} $\sgn(f(m'+\omega))=\sgn(f(b))$,
    \textbf{return} $([m',m'+\omega],N^2)$
  \State \label{eqir:evaluation_end}\textbf{if} $s=\sgn(f(b))$ \textbf{and} $\sgn(f(m'-\omega))=\sgn(f(a))$, \textbf{return} $([m'-\omega,m'],N^2)$
  \State Otherwise, \textbf{return} $(I$,$\sqrt{N})$.
\EndProcedure
\end{algorithmic}
\end{algorithm}

Abbott's QIR method~\cite{abbott-quadratic, kerber-complexity} 
is a hybrid of 
the simple (but inefficient) bisection method with a quadratically converging
variant of the \emph{secant method}. We refer to this method as \textsc{Eqir},
where ``E'' stands for ``exact'' in order to distinguish from the variant
presented in Section~\ref{sec:aqir}.\footnote{To avoid confusion,
the approximate version presented later is also ``exact'' in the sense
that the refined intervals are isolating, but the intermediate
computations are only approximate.}
Given an isolating interval $I=(a,b)$ for a real root $\xi$ of $f$, 
we consider the
secant through $(a,f(a))$ and $(b,f(b))$ (see also Figure~\ref{fig:qir}). This secant intersects the real axis
in the interval $I$, say at $x$-coordinate $m$. 
For $I$ small enough, the secant should approximate the graph of the
function above $I$ quite well and, so,
$m\approx\xi$ should hold. An \textsc{Eqir} step tries to exploit this
fact:

The isolating interval $I$ is (conceptually) subdivided into $N$
subintervals of same size, using $N+1$ equidistant grid points. 
Each subinterval has width $\omega:=\frac{w(I)}{N}$.
Then $m'$,
the closest grid point to $m$, is computed and the sign of $f(m')$ is evaluated.
If that sign equals the sign of $f(a)$, the sign of $f(m'+\omega)$ is
evaluated. Otherwise, $f(m'-\omega)$ is evaluated. If the sign changes
between the two computed values, the interval $(m',m'+\omega)$
or the interval $(m'-\omega,m')$, respectively, 
is set as new isolating interval for $\xi$.
In this case, the \textsc{Eqir} step is called \emph{successful}. 
Otherwise, the isolating interval remains unchanged, and the \textsc{Eqir} step
is called \emph{failing}. See Algorithm~\ref{alg:eqir} for a description in
pseudo-code.

In~\cite{kerber-complexity}, the root refinement problem is analyzed using the
just described \textsc{Eqir} method for the case of integer coefficients and
exact arithmetic with rational numbers.
For that, a sequence of \textsc{Eqir} steps
is performed with $N=4$ initially. After a successful \textsc{Eqir} step,
$N$ is squared for the next step; after a failing step, $N$ is set to 
$\sqrt{N}$. If $N$ drops to $2$, a bisection step is performed,
and $N$ is set to $4$ for the next step. In~\cite{kerber-complexity}, a bound
on the size of an interval is provided to guarantee success of every \textsc{Eqir} and, thus, quadratic convergence of the overall method. 

\section{Approximate QIR}

The most important numerical operation in an \textsc{Eqir} step is the computation
of $f(x_0)$ for values $x_0\in I$. Note that $f(x_0)$ is needed
for determining the closest grid point $m'$ to the secant 
(Step~\ref{eqir:point_location} of Algorithm~\ref{alg:eqir}),
and its sign is required for checking for sign changes in
subintervals (Steps~\ref{eqir:evaluation_begin}-\ref{eqir:evaluation_end}). 

What are the problems if $f$ is a bitstream polynomial as in~(\ref{polyf}), so that $f(x_0)$
can only be evaluated up to a certain precision? First of all, $\frac{Nf(a)}{f(a)-f(b)}$ can only be computed
approximately, too, which might lead to checking the wrong subinterval in the algorithm
if $m$ is close to the center of a subinterval.
Even more seriously, if $f(x_0)$ is zero, then, in general, its sign can never be evaluated
using any precision. Even if we exclude this case, the evaluation of $f(x_0)$
can become costly if $x_0$ is too close to a root of $f$. The challenge is to modify the
QIR method such that it can cope with the uncertainties in the evaluation of $f$,
requires as few precision as possible in a refinement step and still
shows a quadratic convergence behavior eventually.

Bisection is a subroutine called in the QIR method if $N=2$;
before we discuss the general case, we first describe our variant of the bisection
in the bitstream context. Note that we face the same problem: 
$f$ might be equal or almost equal to zero at $\operatorname{mid}(I)$, 
the center of $I$.
We will overcome this problem by evaluating $f$ at several $x$-coordinates
``in parallel''. For that, we subdivide $I$ into 4 equally wide parts
using the subdivision points $m_j:=a+j\cdot \frac{b-a}{4}$ for $1\leq j\leq 3$. 
We also assume that the sign of $f$ at $a$ is already known.
We choose a starting precision 
$\rho$ and compute $f(m_1),\ldots,f(m_3)$ using interval arithmetic in
precision $\rho$ (cf. Section~\ref{apxivarithemtic} for details). 
If less than $2$ out of $3$ signs have been determined using precision $\rho$, 
we set $\rho\gets 2\rho$ and 
repeat the calculation with increased precision. 
Once the sign at at least $2$ subdivision points is determined, 
we can determine a subinterval of at most half the
size of $I$ that contains $\xi$ (Algorithm~\ref{alg:abisect}). 
We will refer to this algorithm as ``bisection'', although the resulting
interval may sometimes be only a quarter of the original size.
Note that $f$ can only become zero at one of the subdivision points which guarantees
termination also in the bitstream context. Moreover, at least $2$ of the $3$ subdivision
points have a distance of at least $\frac{b-a}{8}$ to $\xi$. This asserts that the function
value at these subdivision points is reasonably large and leads to an upper bound
of the required precision (Lemma~\ref{lem:apxbisection}).

\label{sec:aqir}

\begin{figure}[t]
\begin{center}
\vspace{-5cm}
\includegraphics[width=10.5cm]{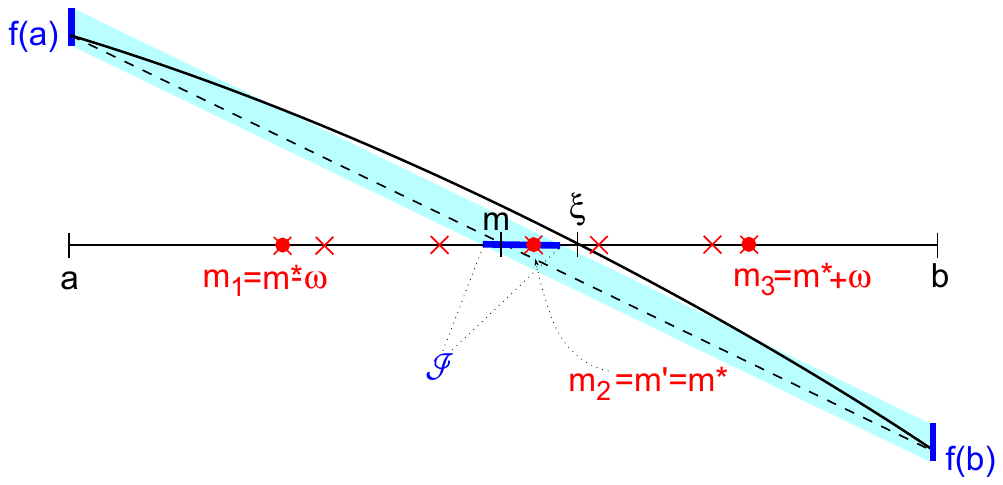}\end{center}
\vspace{-5.75cm}
\caption{\label{fig:qir}Illustration of an \textsc{Aqir} step for $N=4$.}
\end{figure}

We next describe our bitstream variant of the QIR method that we call 
\emph{approximate quadratic interval refinement}, or \textsc{Aqir} for short (see also Figure~\ref{fig:qir} for the illustration of an \textsc{Aqir} step for $N=4$).
Compared to the exact variant, we replace two substeps. In Step~\ref{eqir:point_location}, we replace the computation of $\lambda:=N\frac{f(a)}{f(a)-f(b)}$
as follows: For a working precision $\rho$, 
we evaluate $f(a)$ and $f(b)$ via interval arithmetic with precision $\rho$ (blue vertical intervals in the above figure) and evaluate $N\frac{f(a)}{f(a)-f(b)}$ with interval arithmetic accordingly (cf. Section~\ref{apxivarithemtic}).
Let $J=(c,d)$ denote the resulting interval (in Figure~\ref{fig:qir}, $\mathcal{I}=a+J\cdot\frac{b-a}{N}$ is the intersection of the stripe defined by the interval evaluations of $f(a)$ and $f(b)$ with the real axis). 
If the width $w(J)$ of $J$ is more than $\frac{1}{4}$,
we set $\rho$ to $2\rho$ and retry. Otherwise, 
let $\ell$ be the integer closest to $\operatorname{mid}(J)$ and set $m^*:=a+\ell\cdot\frac{b-a}{N}$.
For $m=a+\frac{f(a)}{f(a)-f(b)}(b-a)$ as before and
$m_j:=a+j\cdot\frac{b-a}{N}$ (red dots) for $j=0,\ldots,N$, 
the following Lemma shows that the computed $m^*=m_\ell$ indeed 
approximates $m$ on the $m_j$-grid:

\begin{algorithm}[t]
\caption{Approximate Bisection}
\label{alg:abisect}
\textsc{Input:} $f\in\R[x]$ square-free, $I=(a,b)$ isolating, $s=\sgn(f(a))$
\newline
\textsc{Output:} $J\subseteq I$ isolating with $2\cdot w(J)\leq w(I)$.
\begin{algorithmic}[1]
\Procedure {Approximate\_Bisection}{$f,I=(a,b),s$}
  \State $V\gets [a+(i-1)\cdot\frac{b-a}{4},i=1,\ldots,5]$
  \State $S=[s,0,0,0,-s]$
  \State $\rho\gets 2$
  \While {$S$ contains more than one zero}
  \For{i=2,\ldots,4}
  \State\label{algstep:ia1} If $S[i]=0$, set $S[i]\gets\sgn\,\IBox(f(V[i]),\rho)$
  \EndFor
  \State $\rho\gets 2\rho$
  \EndWhile
  \State Find $v,w$, such that $S[v]\cdot S[w]=-1\wedge (v+1=w\vee (v+2=w\wedge S[v+1]=0))$
  \State \textbf{return} $(V[v],V[w])$ 
\EndProcedure
\end{algorithmic}
\end{algorithm}

\begin{lemma}\label{lem:point_location}
Let $m$ be inside the subinterval $[m_j,m_{j+1}]$. Then,
$m^*=m_j$ or $m^*=m_{j+1}$.
Moreover, let $m'\in\{m_j,m_{j+1}\}$ be the point that is closer to $m$. 
If $|m-m'|<\frac{b-a}{4N}$, then $m^*=m'$.
\end{lemma}
\begin{proof}
Let $\lambda:=N\frac{f(a)}{f(a)-f(b)}$ and $J$ the interval computed by interval arithmetic as above,
with width at most $\frac{1}{4}$. Since $m=f(a)+\lambda\frac{b-a}{N}\in[m_j,m_{j+1}]$, it follows that $j\le \lambda\le j+1$. By construction, $\lambda\in J$. Therefore, $|\lambda-\operatorname{mid}(J)|\leq\frac{1}{8}$ and, thus, it follows that $\operatorname{mid}(J)$ can only be rounded to $j$ or $j+1$.
Furthermore, for $m'=m_j$, $|m-m'|<\frac{b-a}{4N}$ implies that $|\lambda-j|<\frac{1}{4}$. It follows that 
$|\operatorname{mid}(J)-j|<\frac{3}{8}$ by triangle inequality, so $\operatorname{mid}(J)$
must be rounded to~$j$. The case $m'=m_{j+1}$ is analogous.
\end{proof}

The second substep to replace in the QIR method is to check for sign changes in subintervals in Steps~\ref{eqir:evaluation_begin}-\ref{eqir:evaluation_end}.
As before, we set $\omega:=w(I)/N$.
Instead of comparing the signs at $m'$ and $m'\pm\omega$, we choose seven
subdivision points (red crosses in Figure~\ref{fig:qir}), namely
\begin{align}
m^*-\omega,m^*-\frac{7\omega}{8},m^*-\frac{\omega}{2},m^*,
m^*+\frac{\omega}{2},m^*-\frac{7\omega}{8},m^*+\omega.\label{7points}
\end{align}
In case that $m^*=a$ or $m^*=b$, we only choose the $4$ points of (\ref{7points}) that lie in $I$.
For a working precision $\rho$, we evaluate the sign of $f$ at all subdivision points
using interval arithmetic.
If the sign remains unknown for more than one point, we set
$\rho$ to $2\rho$ and retry. After the sign is determined for all except
one of the points, we look for a sign change in the sequence.
If such a sign change occurs, we set the corresponding interval $I^*$ as isolating
and call the \textsc{Aqir} step \emph{successful}. Otherwise, we call the step \emph{failing}
and keep the old isolating interval.
As in the exact case, we square up $N$ after a successful step,
and reduce it to its square root after a failing step.
See Algorithm~\ref{alg:aqir} for a complete description.

Note that, in case of a successful step, the new isolating interval $I^*$ 
satisfies $\frac{1}{8N}w(I)\leq w(I^*)\leq\frac{1}{N}w(I)$.
Also, similar to the bisection method, the function can only be zero at one of the chosen
subdivision points, and the function is guaranteed to be reasonably large for all but one
of them, which leads to a bound on the necessary precision (Lemma~\ref{lem:apxqir}).
The reader might wonder why we have chosen a non-equidistant grid
involving the subdivision points $m^*\pm\frac{7}{8}\omega$.
The reason is that these additional points allow us to 
give a success guarantee of the method under certain assumptions in the following lemma which is the basis to prove quadratic convergence 
if the interval is smaller than a certain threshold (Section~\ref{sec:aqir_sequence}).\\
\begin{lemma}\label{aqir-success}
Let $I=(a,b)$ be an isolating interval for some root $\xi$ of $f$, 
$s=\sgn(f(a))$ and 
$m$ as before. If $|m-\xi|<\frac{b-a}{8N}=\frac{\omega}{8}$,
then \textsc{Aqir}($f,I,N,s$) succeeds.\\
\end{lemma}
\begin{proof}
Let $m^*$ be the subdivision point selected by the \textsc{Aqir} method.
We assume that $m^*\notin\{a,b\}$; otherwise, a similar (simplified) argument
applies.
By Lemma~\ref{lem:point_location},
$m\in [m^*-\frac{3}{4}\omega,m^*+\frac{3}{4}\omega]$ and, thus, $\xi\in (m^*-\frac{7}{8}\omega,m^*+\frac{7}{8}\omega)$.
It follows that the leftmost two points of (\ref{7points}) have
a different sign than the rightmost two points of (\ref{7points}).
Since the sign of $f$ is evaluated for at least one value on each side, the
algorithm detects a sign change and, thus, succeeds.
\end{proof}

\begin{algorithm}[t]
\caption{Approximate Quadratic interval refinement}
\label{alg:aqir}
\textsc{Input:} $f\in\R[x]$ square-free, $I=(a,b)$ isolating, $N=2^{2^i}\in\N$, $s=\sgn(f(a))$
\newline
\textsc{Output:} $(J,N')$ with $J\subseteq I$ isolating and $N'\in\N$
\begin{algorithmic}[1]
\Procedure {\textsc{Aqir}}{$f,I=(a,b),N$}
  \State \textbf{if} $N=2$, \textbf{return} $($\Call{Approximate\_Bisection}{$f,I,s$},$4)$.
  \State $\omega\gets\frac{b-a}{N}$
  \State $\rho\gets 2$
  \State \label{aqir:point_location}\label{algstep:ia2}
         \textbf{while} $J\gets\IBox(N\frac{f(a)}{f(a)-f(b)},\rho)$ has width
         $>\frac{1}{4}$, set $\rho\gets 2\rho$
  \State $m^*\gets a+\mathrm{round}(\operatorname{mid}(J))\cdot\omega$
  \State \textbf{if} $m^*=a$, $s\gets 4, V\gets [m^*,m^*+\frac{1}{2}\omega,m^*+\frac{7}{8}\omega,m^*+\omega], S\gets [s,0,0,0]$
  \State \textbf{if} $m^*=b$, $s\gets 4, V\gets [m^*-\omega,m^*-\frac{7}{8}\omega,m^*-\frac{1}{2}\omega,m^*], S\gets [0,0,0,-s]$
  \State \textbf{if} $a<m^*<b$, $s\gets 7, V\gets [m^*-\omega,m^*-\frac{7}{8}\omega,m^*-\frac{1}{2}\omega,m^*,m^*+\frac{1}{2}\omega,m^*+\frac{7}{8}\omega,m^*+\omega], S\gets [0,0,0,0,0,0,0]$
  \State $\rho\gets 2$
  \While{$S$ contains more than one zero}\label{aqir:evaluation_begin}
  \For{i=1,\ldots,s}
  \State \label{algstep:ia3} If $S[i]=0$, set $S[i]\gets\sgn\,\IBox(f(V[i]),\rho)$
  \EndFor
  \State $\rho\gets 2\rho$\label{aqir:evaluation_end}
  \EndWhile
  \State \textbf{If} $\exists v,w: S[v]\cdot S[w]=-1\wedge (v+1=w\vee (v+2=w\wedge S[v+1]=0))$
    \textbf{return} $((V[v],V[w]),N^2)$
  \State \textbf{Otherwise, return} $(I,\sqrt{N})$
\EndProcedure
\end{algorithmic}
\end{algorithm}

\section{Analysis of an AQIR step}\label{apxivarithemtic}

The running time of an \textsc{Aqir} step depends on the maximal precision 
$\rho$ needed in the two while loops (Step~\ref{aqir:point_location}, 
Steps~\ref{aqir:evaluation_begin}-\ref{aqir:evaluation_end}) of Algorithm~\ref{alg:aqir}. 
The termination criterion of both loops is controlled by evaluations
of the form $\IBox(E,\rho)$, where $E$ is some polynomial expression
and $\rho$ is the current working precision.

We specify recursively what we understand 
by evaluating $E$ in precision $\rho$ with interval arithmetic.
For that, we define $\down(x,\rho)$ for $x\in\R$ and $\rho\in\N$ to be the
maximal $x_0\leq x$ such that $x_0=\frac{k}{2^\rho}$ for some integer $k$.
The same way $\up(x,\rho)$ is the minimal $x_0\geq x$ with $x_0$ of the same form.
We extend this definition to arithmetic expressions by the following rules
(we leave out $\rho$ for brevity):
\begin{eqnarray*}
\down(E_1+E_2)&:=&\down(E_1)+\down(E_2)\\
\up(E_1+E_2)&:=&\up(E_1)+\up(E_2)\\
\down(E_1\cdot E_2)&:=&\down(\min\{\down(E_1)\down(E_2),\up(E_1)\up(E_2),\\
&&\up(E_1)\down(E_2),\down(E_1)\up(E_2)\})\\
\up(E_1\cdot E_2)&:=&\up(\max\{\down(E_1)\down(E_2),\down(E_1)\up(E_2),\\
&&\up(E_1)\down(E_2),\up(E_1)\up(E_2)\})\\
\down(1/E_1)&:=&\down(1/\up(E_1))\\
\up(1/E_1)&:=&\up(1/\down(E_1))
\end{eqnarray*}
Finally, we define the interval $\IBox(E,\rho):=[\down(E,\rho),\up(E,\rho)]$.
By definition, the exact value of $E$ is guaranteed to be contained
in $\IBox(E,\rho)$. 
We assume that polynomials $f\in\R[x]$ are evaluated according to the Horner scheme, and when evaluating $f(c)$ with precision $\rho$, the above rules apply in each arithmetic step. The next lemma provides a worst case bound on the size of the resulting interval $\IBox(f(c),\rho)$ under certain conditions. We further remark that, in an actual implementation, $\IBox(E,\rho)$ is usually much smaller than the worst case bound derived here. Nevertheless, our complexity analysis is based on the latter bound. Throughout the following considerations , $\Gamma\in\mathbb{N}$ denotes an integer upper bound on the root bound $\Gamma_f$, that is, $\Gamma\ge\Gamma_f$, and, in particular $\log|z_i|\le\Gamma$ for all roots $z_i$ of $f$.\\

\begin{lemma}\label{intarithmetic}
Let $f$ be a polynomial as in (\ref{polyf}), $c\in\R$
with $|c|\leq 2^{\RB+2}$, and $\rho\in\N$.
Then, 
\begin{eqnarray}
|f(c)-\down(f(c),\rho)|\leq 2^{-\rho+1}(d+1)^{2}2^{\tau+ d(\RB+2)}\\
|f(c)-\up(f(c),\rho)|\leq 2^{-\rho+1}(d+1)^{2}2^{\tau+d(\RB+2)}
\end{eqnarray}
In particular, $\IBox (f(c),\rho)$ has a width of at most
$2^{-\rho+2}(d+1)^{2}2^{\tau+d(\RB+2)}$.\\
\end{lemma}
\begin{proof}
We do induction on $d$. The statement is clearly true for $d=0$.
For $d>0$, 
we write $f(c)=a_0+cg(c)$ with $a_0\in\R$ the constant coefficient of $f$ and $g$ of degree $d-1$.
Note that, for any real value $x$, 
$|\down(x,\rho)-x|< 2^{-\rho}$, same for $\up$. Therefore, we can bound
as follows (again, leaving $\rho$ out for simplicity):
\begin{align*}
|f(c)-\down(f(c))| &= |a_0+cg(c)-\down(a_0+cg(c))|= |a_0+cg(c)-\down(a_0)-\down(cg(c))|\\
&\leq |cg(c)-\down(cg(c))| + 2^{-\rho}
\end{align*}
Note that $\down(c\cdot g(c))=\down(H_1(c)\cdot H_2(g(c)))$ 
where $H_{1,2}=\down$ or $H_{1,2}=\up$.
Moreover, we can write $H_1(c)=c-\eps$ with $|\eps|<2^{-\rho}$.
Therefore, we can rearrange
\begin{eqnarray*}
&& |cg(c)-\down(cg(c))| + 2^{-\rho}\leq |cg(c)-(c-\eps)\cdot H_2(g(c))| +2^{-\rho+1}\\
&\leq& |cg(c)-c\cdot H_2(g(c))| + |\eps| \cdot |H_2(g(c))| + 2^{-\rho+1}\\
&\leq& |c|\cdot|g(c)-H_2(g(c))| + 2^{-\rho} |H_2(g(c))| +  2^{-\rho+1}
\end{eqnarray*}
By a simple inductive proof on the degree, we can show that
both $|\up(g(c))|$ and $|\down(g(c))|$ are bounded by
$d 2^{\tau+d(\RB+2)}$. Using that and the induction hypothesis yields
\begin{eqnarray*}
&& |c|\cdot|g(c)-h(g(c))| + 2^{-\rho} |H_2(g(c))| +  2^{-\rho+1}\\
&<& 2^{\RB+2}2^{-\rho+1}d^{2}2^{\tau+(d-1)(\RB+2)}+ 2^{-\rho}d2^{\tau+d(\RB+2)} +  2^{-\rho+1}\\
&\leq& 2^{-\rho+1}(d^2+d+1)2^{\tau+d(\RB+2)}\leq 2^{-\rho+1}(d+1)^{2}2^{\tau d}
\end{eqnarray*}
The bound for $|f(c)-\up(f(c))|$ follows in the same way.
\end{proof}\vspace{0.25cm}

For the sake of simplicity, we decided to assume
fixed-point arithmetic, that means, $\rho$ determines the number
of bits \emph{after the binary point.} 
We refer the interested reader to~\cite[Thm.~12]{ms-cpgeneral2011},
where a corresponding result for floating-point arithmetic is given.

We analyze the required working precision of approximate bisection
and of an \textsc{Aqir} step next. We exploit that, whenever we evaluate
$f$ at $t$ subdivision points, $t-1$ of them have a certain minimal distance
to the root in the isolating interval. The following lemma gives
a lower bound on $|f(x_0)|$ for such a point $x_0$, given that it is
sufficiently far away from any other root of $f$.\\

\begin{lemma}\label{lem:sizeoff}
Let $f$ be as in (\ref{polyf}), $\xi=z_{i_0}$ a real root of $f$ and $x_0$ be a real value with distance $|x_0-z_i|\ge \frac{\sigma_i}{4}$ to all \emph{real} roots $z_i\neq z_{i_0}$. Then,
\[
|f(x_0)|>|\xi-x_0|\cdot 2^{-(2d+\RB+\Sig)}.
\]
(recall the notations from Section~\ref{sec:intro} for the definitions of $\sigma_i$ and $\Sig$)\\
\end{lemma}
\begin{proof}
For each non-real root $z_i$ of $f$, there exists a complex conjugate root $\bar{z}_i$ and, thus, we have $|x_0-z_i|\ge\operatorname{Im}(z_i)\ge\frac{\sigma_i}{2}>\frac{\sigma_i}{4}$ for all $i=m+1,\ldots,d$ as well. It follows that
\begin{align*}
|f(x_0)|&=|a_d\prod_{i=1}^d (x_0-z_i)|=|a_d|\cdot |\xi-x_0|\cdot\prod_{i=1,\ldots,d:i\neq i_0}|x_0-z_i|\\
&\ge |\xi-x_0|\cdot\frac{4}{\sigma_{i_0}}\cdot\prod_{i=1}^d \frac{\sigma_i}{4}> |\xi-x_0|\cdot 2^{-2d-\RB}\cdot2^{-\Sig}, 
\end{align*}
where the last inequality uses that $|z_i|\leq 2^{\RB}$ and, thus, $\sigma(z_{i})\le 2^{\RB+1}$.
\end{proof}

We next analyze an approximate bisection step.\\

\begin{lemma}\label{lem:apxbisection}
Let $f$ be a polynomial as in (\ref{polyf}), $I=(a,b)\subset (-2^{\RB+2},2^{\RB+2})$ be an isolating 
interval for a root $\xi=z_{i_0}$ of $f$ and $s=\sgn(f(a))$. Then, Algorithm~\ref{alg:abisect} 
applied on $(f,I,s)$ requires a maximal precision of
\begin{align*}
\rho_0&:=2\log (b-a)^{-1}+4\log(d+1)+8d+10+2(d+1)\RB+\tau+2\Sig\\
&=O(\log(b-a)^{-1}+\tau+d\RB+\Sig),
\end{align*}
and its bit complexity is bounded by $\tilde{O}(d(\log(b-a)^{-1}+\tau+d\RB+\Sig))$.\\
\end{lemma}

\begin{proof}
Consider the three 
subdivision points $m_j:=a+j\cdot\frac{b-a}{4}$, where $1\le j \le 3$, and an 
arbitrary real root $z_i\neq \xi$ of $f$. 
Note that $|m_j-z_i|>\frac{b-a}{4}$ because the segment from $m_j$ to $z_i$ 
spans at least a quarter of $(a,b)$.
Moreover, $|\xi-m_j|\leq \frac{3}{4}(b-a)$, and so
$$\sigma_i\leq |\xi-z_i| \leq |\xi-m_j| + |m_j-z_i| \leq \frac{3}{4}(b-a)+|m_j-z_i| \leq 4|m_j-z_i|.$$
It follows that $m_j$ has a distance to $z_i$ of at least $\frac{\sigma_i}{4}$.
Hence, we can apply Lemma~\ref{lem:sizeoff} to each $m_j$, that is, we have $|f(m_j)|>|\xi-m_j|\cdot 2^{-(2d+\RB+\Sig)}$. Since the signs of $f$ at the endpoints of $I$ are known, it suffices to compute the signs of $f$ at two of the three subdivision points.  For at least two of these points, the distance of $m_j$ to $\xi$ is at least $\frac{b-a}{8}$, thus, we have $|f(m_j)|>|b-a|\cdot 2^{-(2d+3+\RB+\Sig)}$ for at least two points. Then, due to Lemma~\ref{intarithmetic}, we can use interval arithmetic with a precision $\rho$ to compute these signs if $\rho$ satisfies
\[
2^{-\rho+2}(d+1)^2 2^{\tau+d(\RB+2)}\le(b-a)\cdot 2^{-(2d+3+\RB+\Sig)},
\]
which is equivalent to $\rho\geq\frac{\rho_0}{2}$.
Since we double the precision in each step, we will eventually succeed with a precision 
smaller than $\rho_0$.
The bit complexity for an arithmetic operation with fixed precision $\rho$ is $\tilde{O}(\rho+d\tau)$. Namely, since the absolute 
value of each subdivision point is bounded by $O(\tau)$, the results in the intermediate steps have magnitude $O(d\tau)$ and we consider
$\rho$ bits after the binary point.
At each subdivision point, we have to perform $O(d)$ arithmetic operations for the computation of $f(m_j)$, thus, the costs for these evaluations are bounded by $\tilde{O}(d(d\tau+\rho))$ bit operations. Since we double the precision in each iteration, the total costs are dominated by the last successful evaluation and, thus, we have to perform $\tilde{O}(d(\rho_0+d\tau))=\tilde{O}(d(\log(b-a)^{-1}+d\tau+\Sig))$ bit operations.
\end{proof}\vspace{0.25cm}

We proceed with the analysis of an \textsc{Aqir} step. 
In order to bound the required precision, we need additional properties of the isolating interval.\\
\begin{definition}\label{def:normal}
Let $f$ be as in (\ref{polyf}), $I:=(a,b)$ be an isolating interval of a root $\xi$ of $f$. We call $I$ \emph{normal}\footnote{The reader may notice that the definition of "normal" depends on the upper bound $\Gamma$ on $\Gamma_f$. Throughout our argument, we assume that such an initial $\Gamma$ is given. We will finally choose a $\Gamma$ which approximates $\Gamma_f$ up to an (addative) error of $O(\log d)$.} if\vspace{0.25cm}
\begin{itemize}
\item $I\subseteq (-2^{\RB+2},2^{\RB+2})$,
\item $|p-z_i|>\frac{\sigma_i}{4}$ for every $p\in I$ and $z_i\neq \xi$, and
\item $\min\{|f(a)|,|f(b)|\}\geq 2^{-(28+2\tau+17d\RB+2\Sig-5\log(b-a))}.$\\
\end{itemize}
\end{definition}
In simple words, a normal isolating interval has a reasonable distance to any other root of $f$,
and the function value at the endpoints is reasonably large.
We will later see that 
it is possible to get normal intervals 
by a sequence of approximate bisection steps.\\  

\begin{lemma}\label{lem:apxqir}
Let $f$ be a polynomial as in (\ref{polyf}), $I=(a,b)$ be a normal isolating 
interval for a root $\xi=z_{i_0}$ of $f$ with $s=\sgn(f(a))$,
and let $N\leq 2^{2(\RB+4-\log(b-a))}$.
Then, the \textsc{Aqir} step for $(f,I,N,s)$ requires a precision of at most
$$\rho_{max}:=87d\tau+17d\RB+4\Sig-14\log(b-a)$$
and, therefore, its bit complexity is bounded by
$$\Otilde(d(\tau+d\RB+\Sig-\log(b-a))).$$
Moreover, the returned interval is again normal.\\
\end{lemma}
\begin{proof}
We have to distinguish two cases. 
For $N>2$, we consider the two while-loops in Algorithm~\ref{alg:aqir}.
In the first loop (Step~\ref{aqir:point_location}), we evaluate $N\frac{f(a)}{f(a)-f(b)}$ 
via interval arithmetic, doubling the precision $\rho$ until the width of the resulting 
interval $J$ is less than or equal to $1/4$.
The following considerations show that we can achieve this if $\rho$ satisfies
\begin{align}
2^{-\rho+2}(d+1)^2 2^{\tau+d(\RB+2)}\le\frac{\min(|f(a)|,|f(b)|)}{32N}.\label{suffprec}
\end{align}
W.l.o.g., we assume $f(a)>0$. If $\rho$ satisfies the above condition, then, due to Lemma~\ref{intarithmetic}, $\IBox(N\cdot f(a),\rho)$ is contained within the interval
$$[Nf(a)-\frac{|f(a)|}{32},Nf(a)+\frac{|f(a)|}{32}]
=Nf(a)\cdot[1-\frac{1}{32N},1+\frac{1}{32N}]
$$ 
and $\IBox(f(a)-f(b),\rho)$ is contained within the interval
\begin{align*}
[f(a)-f(b)-\frac{|f(a)-f(b)|}{32N},f(a)-f(b)+\frac{|f(a)-f(b)|}{32N}]
=(f(a)-f(b))\cdot[1-\frac{1}{32N},1+\frac{1}{32N}],
\end{align*} where the latter result uses the fact that $f(a)$ and $f(b)$ have different signs.
It follows that $\IBox(N\frac{f(a)}{f(a)-f(b)},\rho)$ is contained within $\frac{Nf(a)}{f(a)-f(b)}\cdot[ (1-\frac{1}{32N})/(1+\frac{1}{32N}),(1+\frac{1}{32N})/(1-\frac{1}{32N})]$, and a simple computation shows that $N\cdot [ (1-\frac{1}{32N})/(1+\frac{1}{32N}),(1+\frac{1}{32N})/(1-\frac{1}{32N})]$ has width less than $1/4$. Hence, since $\frac{f(a)}{f(a)-f(b)}$ has absolute value less than $1$, $\IBox(N\frac{f(a)}{f(a)-f(b)},\rho)$ has width less than $1/4$ as well. The bound (\ref{suffprec}) on $\rho$ also writes as
\begin{align*}
\rho&\ge 7+2\log(d+1)+\tau+d\RB+2d+\log N+\log\min(|f(a),f(b)|)^{-1}
\end{align*}
and since we double $\rho$ in each iteration, computing $N\frac{f(a)}{f(a)-f(b)}$ via interval arithmetic up to an error of $1/4$ demands for a precision 
\begin{align*}
\rho&< 14+4\log(d+1)+2\tau+2d\RB+4d+2\log N+2\log\min(|f(a),f(b)|)^{-1}\\
&< 14+2\tau+10d\RB+2\log N+ 2\log\min(|f(a),f(b)|)^{-1},
\end{align*}
Since $I$ is normal and because of the posed condition on $N$, we can bound this by
\begin{align*}
\rho&<11d\tau+4(\tau+5-\log(b-a))+2(32d\tau+2\Sig-5\log(b-a))\\
&<87d\tau+4\Sig-14\log(b-a)<\rho_{max}.
\end{align*}
We turn to the second while loop of Algorithm~\ref{alg:aqir} (Steps~\ref{aqir:evaluation_begin}-\ref{aqir:evaluation_end})
where $f$ is evaluated at the subdivision points 
$m^*-\omega,m^*-\frac{7\omega}{8},\ldots,m^*+\omega$ as defined in (\ref{7points}). 
Since the interval is normal, we can apply Lemma~\ref{lem:sizeoff} to each of the seven subdivision points. 
Furthermore, at least six of these points have distance $\ge \frac{b-a}{16N}$ to the root $\xi$ and, thus, 
for these points, $|f|$ is larger than $\frac{b-a}{16N}\cdot 2^{-(2d+\tau+\Sigma_f)}$. 
Then, according to Lemma~\ref{suffprec}, it suffices to use a precision $\rho$ that fulfills
$$
2^{-\rho+2}(d+1)^2 2^{\tau+d(\RB+2)}\leq \frac{b-a}{16N}\cdot 2^{-(2d+\RB+\Sig)},\text{ or}
$$
$$\rho\ge\rho_1:=6+2\log(d+1)+\tau+d\RB+4d+\RB+\Sig+\log N-\log(b-a).$$
The same argumentation as above then shows that the point evaluation will be performed with a 
maximal precision of less than 
\begin{align*}
2\rho_1&<2(6+\tau+7d\RB+\RB+\Sig+\log N-\log(b-a))\\
&\leq 12+2\tau+14d\RB+2\RB+2\Sig+4(\RB+4-\log(b-a))-\log(b-a)\\
&\leq 28+2\tau+17d\RB+2\Sig-5\log(b-a)
\end{align*}
which is bounded by $\rho_{max}$. 
Moreover, at the new endpoints $a'$ and $b'$, $|f|$ is at least
$$2^{-2\rho_1}\geq 2^{-(28+2\tau+17d\RB+2\Sig-5\log(b-a))}\geq 2^{-(28+2\tau+17d\RB+2\Sig-5\log(b'-a'))}$$
which proves that $I'=(a',b')$ is again normal.

It remains the case of $N=2$, where a bisection step is performed. It is straight-forward
to see with Lemma~\ref{lem:apxbisection} that the required precision is bounded by $\rho_{max}$,
and in an analogue way as for the point evaluations for $N>2$, we can see that the resulting
interval is again normal.
By the same argument as in Lemma~\ref{lem:apxbisection},
the overall bit complexity of the \textsc{Aqir} step is bounded by
$$\Otilde(d\rho_{max})=\Otilde(d(d\tau+\Sig-\log(b-a))).$$
\end{proof}

\section{Root refinement}
\label{sec:root_refinement}

We next analyze the complexity of our original problem:
Given a polynomial $f$ as in (\ref{polyf}) and isolating
intervals for all its real roots, refine
the intervals to a size of at most $2^{-L}$.
Our refinement method consists of two steps. First, we turn
the isolating intervals into normal intervals by applying
bisections repeatedly.
Second, we call the \textsc{Aqir} method repeatedly on the intervals
until each has a width of at most $2^{-L}$.
Algorithm~\ref{alg:main} summarizes our method for root refinement.
We remark that depending on the properties of the root isolator 
used to get initial isolating intervals,
the normalization can be skipped; this is for instance
the case when using the isolator from~\cite{sag-complexity}. 
We also emphasize that the normalization is unnecessary for the correctness
of the algorithm; its purpose is to prevent the working precision
in a single \textsc{Aqir} step of growing too high.

\subsection{Normalization}
\label{sec:normalization}

If there exists only one isolating interval, it is easily shown that $(-2^{\Gamma+2},2^{\Gamma+2})$ is already a normal interval that isolates the corresponding root. Hence
we assume that at least two isolating intervals are present.
The normalization (Algorithm~\ref{alg:normal}) 
consists of two steps: first, the isolating
intervals are refined using approximate bisection
until the distance between two consecutive intervals is at least
three times larger than the size of the larger of the two
involved intervals. This ensures that 
all points in an isolating interval are reasonably far away
from any other root of $f$.
In the second step, each interval
is enlarged on both sides by an interval of at least the same size as itself.
This ensures that the endpoints are sufficiently far away
from any root of $f$ to prove a lower bound of $f$
at the endpoints.
W.l.o.g., we also assume that the input intervals are contained
in $(-2^{\Gamma+1},2^{\Gamma+1})$ because
all roots are contained in that interval, 
so the leftmost and rightmost intervals can just be cut if necessary. 
Obviously, the resulting intervals are still isolating 
and disjoint from each other.
Moreover, they do not become too small during the bisection process:\\

\begin{algorithm}[t]
\caption{Normalization}
\label{alg:normal}
\textsc{Input:} $f\in\R[t]$ a polynomial as in (\ref{polyf}), $I_1=(a_1,b_1),\ldots,I_m=(a_m,b_m)$ disjoint isolating intervals in ascending order, $m\geq 2$, $s_1,\ldots,s_m$ with $s_k=\sgn(f(\min I_k))$
\newline
\textsc{Output:} normal isolating intervals $J_1,\ldots,J_m$ with $z_k\in I_k\cap J_k$
\begin{algorithmic}[1]
\Procedure {normalize}{$f,I_1,\ldots,I_m$}
  \For{k=1,\ldots,m-1}
    \While{$\min I_{k+1}-\max I_{k} < 3 \max\{w(I_k),w(I_{k+1})\}$}
      \State \textbf{if} {$w(I_k)>w(I_{k+1})$}
      \State \textbf{then} \Call{Approximate\_bisection}{$f,I_k,s_k$}
      \State \textbf{else} \Call{Approximate\_bisection}{$f,I_{k+1},s_{k+1}$} 
    \EndWhile
    \State $d_k\gets \min I_{k+1}-\max I_{k}$
  \EndFor
 
  \State $d_0\gets d_1$, $d_m\gets d_{m-1}$
  \For{k=1,\ldots,m}
    \State $J_k\gets (a_k-d_{k-1}/4,b_k+d_k/4)$ \Comment{enlarge $I_k$ by more than $w(I_k)$ at both sides}
  \EndFor 
  \State \textbf{return} $J_1,\ldots,J_m$
\EndProcedure
\end{algorithmic}
\end{algorithm}

\begin{lemma}\label{lem:normal_lower_bound}
For $J_1,\ldots,J_m$ as returned by Alg.~\ref{alg:normal},
$w(J_k)>\frac{1}{4}\sigma_k$.\\
\end{lemma}
\begin{proof}
After the first for-loop, the distance $d_k$ 
between any two consecutive intervals $I_k$ and $I_{k+1}$ satisfies $d_k\ge 
3\max\{w(I_k), w(I_{k+1})\}$, thus $\sigma_k<w(I_k)+w(I_{k+1})+d_k<2d_k$. 
Hence, in the last step, each $I_k$ is enlarged by \emph{more than} 
$\sigma_k/8$ on each side. This proves that the corresponding enlarged intervals $J_k$ have size 
\emph{more than} $\sigma_k/4$. 
\end{proof}\vspace{0.25cm}

\begin{lemma}\label{lem:normalization}
Algorithm~\ref{alg:normal} is correct, i.e., returns normal intervals.\\
\end{lemma}
\begin{proof}
Let $J_1,\ldots,J_m$ denote the returned intervals,
and fix some interval $J_k$ containing the root $z_k$
of $f$. We have to prove the three properties of Definition~\ref{def:normal}.
The first property is clear because the initial interval are assumed to lie in $(-2^{\Gamma+1},2^{\Gamma+1})$, and they are extended by not more than $2^\Gamma$
to each side.
In the proof of Lemma~\ref{lem:normal_lower_bound}, we have already shown that $I_k$ is eventually enlarged by more than $\sigma_k/8$ on each side. More precisely, the right endpoint of $J_k$ has distance at least $d_k/4>\sigma_{k+1}/8$ to $J_{k+1}$, and the left endpoint of $J_k$ has distance  at least $d_{k-1}/4>\sigma_{k-1}/8$ to $J_{k-1}$. It follows that, for each $x_0\in J_k$, we have $|x_0-z_{k\pm 1}|< \sigma_{k\pm 1}/4$, respectively. Hence, the second property in Definition~\ref{def:normal} is fulfilled. 
For the third property of Definition~\ref{def:normal}, 
let $e$ be one of the endpoints of $J_k$. 
We have just proved that the distance to every root $z_i$ except $z_k$ is 
more than $\frac{\sigma_i}{4}$ and $|e-z_k|>\sigma_k/8$. With an estimation similar as in the proof of Lemma~\ref{lem:sizeoff}, we obtain:
\begin{eqnarray*}
&&|f(e)| > \frac{\sigma_k}{8}\prod_{i\neq k} \frac{\sigma_i}{4}= \frac{1}{8}\cdot\frac{1}{4^{d-1}}2^{-\Sigma_f}
= 2^{-(2d+\Sig+1)},
\end{eqnarray*}
and $2^{-(2d+\Sig+1)}\geq 2^{-(28+2\tau+17d\Gamma+2\Sig-5\log(b-a))}$
because $\log(b-a)\leq \Gamma+2$ and $-\Sig\leq d(\Gamma+1)<2d\Gamma$.
\end{proof}\vspace{0.25cm}

\begin{lemma}\label{lem:normalization_cost}
Algorithm~\ref{alg:normal} has a complexity of
$$\Otilde(d(d\RB + \Sigma_f)(\tau+d\RB+\Sigma_f))$$
\end{lemma}
\begin{proof}
As a direct consequence of Lemma~\ref{lem:normal_lower_bound},
each interval $I_k$ is only bisected $O(\RB+\log(\sigma_k)^{-1})$ 
many times because each starting interval is assumed to be contained in $(-2^{\RB+1},2^{\RB+1})$. So the total number of bisections adds up
to $O(d\RB+\Sigma_f)$ considering all roots of $f$. Also, the size of the isolating interval $I_k$
is lower bounded by $\frac{3}{20}\cdot\sigma_k=2^{-O(\Sig+d\RB)}$, so that one approximate bisection 
step has a complexity of $\Otilde(d(\tau+d\RB + \Sigma_f))$ due to Lemma~\ref{lem:apxbisection}.
\end{proof}\vspace{0.25cm}

\begin{algorithm}[t]
\caption{Root Refinement}
\label{alg:main}
\textsc{Input:} $f=\sum a_ix^i\in\R[t]$ a polynomial as in (\ref{polyf}), isolating intervals $I_1,\ldots, I_m$
for the real roots of $f$ in ascending order, $L\in\Z$
\newline
\textsc{Output:} isolating intervals $J_1,\ldots,J_m$ with $w(J_k)\leq 2^{-L}$
\begin{algorithmic}[1]
\Procedure {root\_refinement}{$f,L,I_1,\ldots,I_m$}
  \State $s_k:=\sgn(a_d)\cdot (-1)^{m-k+1}$ \Comment{$s_k=\sgn(f(\min I_k))$} 
  \State $J_1,\ldots,J_m\gets\Call{normalize}{f,I_1,\ldots,I_m}$
  \For{k=1,\ldots,m}
    \State $N\gets 4$
    \State \textbf{while} $w(J_k)>2^{-L}$ \textbf{do} $(J_k,N)\gets$\Call{\textsc{Aqir}}{$f,J_k,N,s_k$}
  \EndFor
  \State \textbf{return} $J_1,\ldots,J_m$
\EndProcedure
\end{algorithmic}
\end{algorithm}

\subsection{The AQIR sequence}
\label{sec:aqir_sequence}

It remains to bound the cost of the calls of \textsc{Aqir}. 
We mostly follow the argumentation from~\cite{kerber-complexity},
mostly referring to that article for technical proofs.
We introduce the following convenient notation:\\
\begin{definition}
Let $I_0:=I$ be a normal isolating interval for some real root $\xi$ of $f$, $N_0:=4$ and $s:=\sgn(\min I_0)$.
The \textsc{Aqir} \emph{sequence} $(S_0,S_1,\ldots,S_{v_{\xi}})$ is defined by
$$S_0:=(I_0,N_0)=(I,4)\quad S_i=(I_i,N_i):=\textsc{Aqir}(f,I_{i-1},N_{i-1},s)\text{ for }i\geq 1,$$
where $v_{\xi}$ is the first index such that the interval $I_{v_\xi}$ has width at most $2^{-L}$.
We say that $S_i\aqirarrow S_{i+1}$ \emph{succeeds} if \textsc{Aqir}($f,I_i,N_i,s$) succeeds,
and that $S_i\aqirarrow S_{i+1}$ \emph{fails} otherwise.
\end{definition}

\noindent As in~\cite{kerber-complexity}, we divide the QIR sequence into two parts
according to the following definition:\\

\begin{definition}\label{def:C_bound}
For $\xi$ a root of $f$, we define

$$C_{\xi}:=\frac{|f'(\xi)|}{8\left( \frac{d^2}{\sigma(\xi,f)}|f'(\xi)| + \sum_{i=2}^{d} \left(\frac{\sigma(\xi,f)}{d^2}\right)^{i-2}|f^{(i)}(\xi)|\right)}.$$
For $(S_0,\ldots,S_{v_\xi})$ the QIR sequence of $\xi$, define
$k$ as the minimal index such that $S_k=(I_k,N_k)\aqirarrow S_{k+1}$
succeeds and $w(I_k)\leq C_{\xi}$.
We call $(S_0,\ldots,S_k)$ \emph{linear sequence} and $(S_k,\ldots,S_{v_\xi})$
\emph{quadratic sequence} of $\xi$.
\end{definition}\vspace{0.25cm}

Note that~\cite{kerber-complexity} defined a different threshold for splitting the QIR sequence,
and the linear sequence was called \emph{initial sequence} therein.
We renamed it to avoid confusion with the initial normalization phase in our variant.\\

\textbf{Quadratic convergence.} We start by justifying the name ``quadratic sequence''.
Indeed, it turns out that all but one \textsc{Aqir} step in the quadratic sequence
are successful, hence, $N$ is squared in (almost) every step and therefore,
the refinement factor of the interval is doubled in (almost) every step.
We first prove two important properties of $C_\xi$ as defined in Defition~\ref{def:C_bound}:\\

\begin{lemma}\label{qir-bound-lemma}
Let $\xi\in\C$ be a root of $f$.\vspace{0.25cm} 
\begin{enumerate}
\item $0<C_\xi\leq\frac{\sigma(\xi,f)}{8d^2}$
\item Let $\mu\in\C$ be such that $|\xi-\mu|<C_\xi$. Then
$$C_\xi < \frac{|f'(\xi)|}{8|f''(\mu)|}.$$\\
\end{enumerate}
\end{lemma}
\begin{proof}
Note that all summands in the denominator of $C_\xi$ are non-negative. Therefore, the first property follows
immediately by removing all but the first summand in the denominator.

For the second property, we consider the Taylor expansion of $f''(\mu)$ in $\xi$:
$$f''(\mu)=\sum_{i=2}^{d} (\mu-\xi)^{i-2} \frac{f^{(i)}(\xi)}{(i-2)!}.$$
Because $|\mu-\xi| < C_\xi < \frac{\sigma(\xi)}{d^2}$ by the first property, we can bound
$$|f''(\mu)| < \sum_{i=2}^{d} \left(\frac{\sigma(\xi)}{d^2}\right)^{i-2} |f^{(i)}(\xi)|.$$
It follows that
$$\frac{|f'(\xi)|}{8|f''(\mu)|} > \frac{|f'(\xi)|}{8\left(\sum_{i=2}^{d} \left(\frac{\sigma(\xi)}{d^2}\right)^{i-2} |f^{(i)}(\xi)|\right)} > C_{\xi}$$
\end{proof}\vspace{0.25cm}

The following bound
follows from considering the Taylor expansion of $f$ at $\xi$
in the expression for $m$:\\
\begin{lemma}\label{lem:sufficientlysmall}~\textbf{\emph{\cite[Thm.~4.8]{kerber-complexity}}}
Let $(a,b)$ be isolating for $\xi$ with width $\delta<C_\xi$ and $m$ as in Lemma~\ref{aqir-success} (i.e., $m=a+\frac{f(a)}{f(a)-f(b)}(b-a)$).
Then, $|m-\xi|\leq \frac{\delta^2}{8C_{\xi}}$.\\
\end{lemma}
\begin{proof}
We consider the Taylor expansion of $f$ at $\xi$. For a given $x\in (a,b)$, 
we have
$$f(x)=f'(\xi)(x-\xi)+\frac{1}{2}f''(\tilde{\xi})(x-\xi)^2$$
with some $\tilde{\xi}\in [x,\xi]$ or $\tilde{\xi}\in[\xi,x]$.
Thus, we can simplify
\begin{align*}
|m-\xi| &= \left|\frac{f(b)(a-\xi)-f(a)(b-\xi)}{f(b)-f(a)}\right|= \left|\frac{\frac{1}{2}(f''(\tilde{\xi}_1)(b-\xi)^2(a-\xi)-f''(\tilde{\xi}_2)(a-\xi)^2(b-\xi))}{f(b)-f(a)}\right|\\
&\leq \frac{1}{2}|b-\xi||a-\xi| \cdot\frac{|f''(\tilde{\xi}_1)|(b-\xi)+|f''(\tilde{\xi}_2)|(\xi-a)}{|f(b)-f(a)|}
\leq \frac{\delta^2\max\{|f''(\tilde{\xi}_1)|,|f''(\tilde{\xi}_2)|\}}{2|f'(\nu)|}
\end{align*}
for some $\nu\in (a,b)$. The Taylor expansion of $f'$ yields
$f'(\nu)=f'(\xi)+f''(\tilde{\nu})(\nu-\xi)$
with $\tilde{\nu}\in (a,b)$. Since $\delta\leq C_\xi$, it follows with
Lemma~\ref{qir-bound-lemma}
$$|f''(\tilde{\nu})(\nu-\xi)|\leq |f''(\tilde{\nu})|C_\xi\leq \frac{1}{8}|f'(\xi)|.$$
Therefore $|f'(\nu)|>\frac{7}{8}|f'(\xi)|>\frac{1}{2}|f'(\xi)|$, and it follows again with
Lemma~\ref{qir-bound-lemma} that
\begin{eqnarray*}\
|m-\xi|\leq \frac{\delta^2\max\{|f''(\tilde{\xi}_1)|,|f''(\tilde{\xi}_2)|\}}{|f'(\xi)|}
\leq \frac{\delta^2}{8\frac{|f'(\xi)|}{8\max\{|f''(\tilde{\xi}_1)|,|f''(\tilde{\xi}_2)|\}}}<\frac{\delta^2}{8C_\xi}.
\end{eqnarray*}
\end{proof}\vspace{0.25cm}

\begin{corollary}\label{cor:all_works}
Let $I_j$ be an isolating interval for $\xi$ of width 
$\delta_j\leq \frac{C_\xi}{N_j}$. Then, each call of the \textsc{Aqir} sequence
$$(I_j,N_j)\aqirarrow(I_{j+1},N_{j+1})\aqirarrow\ldots$$ succeeds.\\
\end{corollary}
\begin{proof}
We use induction on $i$. Assume that the first $i$ \textsc{Aqir} calls succeed. Then, 
another simple induction shows that 
$\delta_{j+i}:=w(I_{j+i})\leq \frac{N_j\delta_j}{N_{j+i}}<\frac{C_\xi}{N_{j+i}}$, where we use that $N_{j+i}=N_{j+i-1}^2$. 
Then, according to Lemma~\ref{lem:sufficientlysmall}, we have that
$$|m-\xi|\leq \delta_{j+i}^2\frac{1}{8C_\xi}\leq \delta_{j+i}\frac{C_\xi}{N_{j+i}}\frac{1}{8C_\xi}=\frac{1}{8}\frac{\delta_{j+i}}{N_{j+i}},$$
with $m$ as above.
By Lemma~\ref{aqir-success}, the \textsc{Aqir} call succeeds.
\end{proof}\vspace{0.25cm}

\begin{corollary}\textbf{\emph{\cite[Cor.~4.10]{kerber-complexity}}}
In the quadratic sequence, there is at most one failing \textsc{Aqir} call.\\
\end{corollary}
\begin{proof}
Let $(I_i,N_i)\aqirarrow (I_{i+1},N_{i+1})$ 
be the first failing \textsc{Aqir} call in the quadratic sequence.
Since the quadratic sequence starts with a successful \textsc{Aqir} call,
the predecessor $(I_{i-1},N_{i-1})\aqirarrow (I_i,N_i)$ 
is also part of quadratic sequence, and succeeds.
Thus we have the sequence
$$(I_{i-1},N_{i-1})\stackrel{Sucess}{\aqirarrow} (I_i,N_i)\stackrel{Fail}{\aqirarrow} (I_{i+1},N_{i+1})$$

One observes easily that 
$w(I_{i+1})=w(I_i)=\frac{w(I_{i-1})}{N_{i-1}}\leq \frac{C_\a}{N_{i-1}}$,
and $N_{i+1}=\sqrt{N_{i}}=\sqrt{N_{i-1}^2}=N_{i-1}$.
By Corollary~\ref{cor:all_works}, all further \textsc{Aqir} calls succeed.
\end{proof}\vspace{0.25cm}

\textbf{Cost of the linear sequence.} 
We bound the costs of refining the isolating interval of $\xi$ to size $C_\xi$
with \textsc{Aqir}.
We first show that, on average, the \textsc{Aqir} sequence 
refines by a factor two in every second step.
This shows in particular that refining using \textsc{Aqir} is at most a factor
of two worse than refining using approximate bisection.\\
\begin{lemma}\label{lem:aqir_not_worse_than_bisection}
Let $(S_0,\ldots,S_\ell)$ denote an arbitrary prefix of the \textsc{Aqir} sequence for $\xi$,
starting with the isolating interval $I_0$ of width $\delta$.
Then, the width of $I_\ell$ is not larger than $\delta 2^{-(\ell-1)/2}$.\\
\end{lemma}
\begin{proof}
Consider a subsequence $(S_i,\ldots,S_{i+j})$ of $(S_0,\ldots,S_\ell)$ such that
$S_i\aqirarrow S_{i+1}$ is successful, but any other step in the subsequence fails.
Because there are $j$ steps in total, and thus $j-1$ consecutive failing steps,
the successful step must have used a $N$ with $N\geq 2^{2^{j-1}}$. Because
$2^{j-1}\geq\frac{j}{2}$, it holds that
$$w(I_{i+j})\leq \frac{w(I_i)}{N}\leq w(I_{i+j})2^{-2^{j-1}}\leq w(I_{i+j})2^{-j/2}.$$
Repeating the argument for maximal subsequences of this form, we get that
either $w(I_\ell)\leq w(I_0)2^{-\ell/2}$ if the sequence starts with a successful step,
or $w(I_\ell)\leq w(I_0)2^{-(\ell-1)/2}$ otherwise, because the second step must
be successful in this case.
\end{proof}\vspace{0.25cm}

We want to apply Lemma~\ref{lem:apxqir} to bound the bit complexity of a single \textsc{Aqir} step.
The following lemma shows that the condition on $N$ from Lemma~\ref{lem:apxqir}
is always met in the \textsc{Aqir} sequence.\\

\begin{lemma}
Let $(I_j,N_j)\aqirarrow(I_{j+1},N_{j+1})$ be a call in an \textsc{Aqir} sequence and $I_j:=(a,b)$.
Then, $N_j\leq 2^{2(\RB+4-\log(b-a))}$.\\
\end{lemma}
\begin{proof}
We do induction on $j$.
Note that $I_0\subset(-2^{\RB+2},2^{\RB+2})$ by normality, 
hence $b-a\leq 2^{\RB+3}$.
It follows that $2^{2(\RB+4-\log(b-a))}\geq 4=N_0$.
Assume that the statement is true for $j-1$. If the previous step
$(I_{j-1},N_{j-1})\aqirarrow(I_{j},N_{j})$ is failing, then
$N_j=\sqrt{N_{j-1}}$ and the isolating interval remains unchanged, so the statement
is trivially correct. If the step is successful, then it holds that
$(b-a)\leq\frac{2^{\RB+3}}{\sqrt{N_j}}$. By rearranging terms,
we get that $N_j\leq 2^{2(\RB+3-\log(b-a))}$.
\end{proof}\vspace{0.25cm}

It follows inductively that the conditions of Lemma~\ref{lem:apxqir} 
are met for each call in the \textsc{Aqir} sequence 
because $I_0$ is normal by construction.
Therefore, the linear sequence for a root $\xi$ of $f$
is computed with a bit complexity of
\begin{eqnarray}
\label{linear_complexity_bound}
\tilde{O}((\RB+\log (C_\xi)^{-1}) d (\log(C_\xi^{-1})+\tau+d\RB+\Sig))
\end{eqnarray}
because $O(\RB+\log (C_\xi^{-1}))$ steps are necessary
to refine the interval to a size smaller than $C_\xi$ 
by Lemma~\ref{lem:aqir_not_worse_than_bisection},
and the bit complexity is bounded
by $\Otilde(d (\log(C_\xi^{-1})+\tau+d\RB+\Sig))$ with Lemma~\ref{lem:apxqir}.
It remains to bound $\log(C_\xi)^{-1}$; we do so by bounding the sum of all
$\log(C_\xi)^{-1}$ with the following lemma.\\

\begin{lemma}\label{lem:sum_of_C_alphas}
$\sum_{i=1}^{m} \log (C_{z_{i}})^{-1} = O(d(\RB+\log d)+\Sig))$\\
\end{lemma}
\begin{proof}
We note that
$$\sum_{\ell=1}^{m} \log(C_{z_{\ell}})^{-1}=\sum_{\ell=1}^{m} \log \left(8\cdot\left(\frac{d^2}{\sigma_\ell}+\sum_{i=2}^d \left(\frac{\sigma_\ell}{d^2}\right)^{i-2}\left|\frac{f^{(i)}(z_\ell)}{f'(z_\ell)}\right|\right)\right).$$
We focus on the quotient $\left|\frac{f^{(i)}(z_\ell)}{f'(z_\ell)}\right|$. Let $z_1',\ldots,z_{d-1}'$ denote the (not necessarily distinct) roots of $f'$.
Note that for $x\in\C$ and any $i\geq 1$, 
$$f^{(i)}(x)=a_d\sum_{\stackrel{X\subseteq\{1,\ldots,n-1\}}{|X|=i-1}}\, \prod_{\stackrel{j\in\{1,\ldots,d-1\}}{j\notin X}} (x-z_j')$$
Therefore, the quotient writes as
$$\left|\frac{f^{(i)}(z_\ell)}{f'(z_\ell)}\right|=\left|\sum_{\stackrel{X\subseteq\{1,\ldots,d-1\}}{|X|=i-1}}\, \prod_{j\in X} \frac{1}{z_\ell-z_j'}\right|\leq \sum_{\stackrel{X\subseteq\{1,\ldots,d-1\}}{|X|=i-1}}\, \prod_{j\in X} \frac{1}{|z_\ell-z_j'|}. $$
Since $|z_\ell-z_j'|\geq \frac{\sigma_\ell}{d}$~\cite[Thm.8]{Eigenwillig2007a}, we can further bound this to
$$\sum_{\stackrel{X\subseteq\{1,\ldots,d-1\}}{|X|=i-1}}\, \prod_{j\in X} \frac{1}{|z_\ell-z_j'|}\leq \sum_{\stackrel{X\subseteq\{1,\ldots,d-1\}}{|X|=i-1}}\, \prod_{j\in X} \frac{d}{\sigma_\ell} \leq \sum_{\stackrel{X\subseteq\{1,\ldots,d-1\}}{|X|=i-1}} \left(\frac{d}{\sigma_\ell}\right)^{i-1} \leq d^{i-1}\left(\frac{d}{\sigma_\ell}\right)^{i-1} = \frac{d^{2i-2}}{\sigma_\ell^{i-1}},$$
and, therefore,
$$\sum_{i=2}^d \left(\frac{\sigma_\ell}{d^2}\right)^{i-2}\left|\frac{f^{(i)}(z_\ell)}{f'(z_\ell)}\right| \leq \sum_{i=2}^d \left(\frac{\sigma_\ell}{d^2}\right)^{i-2}\frac{d^{2i-2}}{\sigma_\ell^{i-1}}=\sum_{i=2}^{d}\frac{d^2}{\sigma_\ell}=(d-1)\frac{d^2}{\sigma_\ell}.$$
Plugging in into the overall sum yields
\begin{eqnarray*}
&&\sum_{\ell=1}^{m} \log(C_{z_{\ell}})^{-1}=\sum_{\ell=1}^{m}\log \left(8\cdot\left(\frac{d^2}{\sigma_\ell}+(d-1)\frac{d^2}{\sigma_\ell}\right)\right)
= 3d+ \sum_{\ell=1}^{m}\log \frac{d^3}{\sigma_\ell}\\
&=& 3d+3m\log d + \Sig + \sum_{\ell=m+1}^{d}\log\sigma_\ell
\leq 3d+3d\log d + \Sig + d(\RB+1)=O(d(\RB+\log d)+\Sig).
\end{eqnarray*}
\end{proof}\vspace{0.25cm}

\begin{lemma}\label{lem:linear_cost}
The linear sequences for all real roots are computed within a total bit complexity of
$$\tilde{O}(d(d\RB+\Sig)(\tau+d\RB+\Sig).\\$$
\end{lemma}
\begin{proof}
The total cost of all linear sequences is bounded by
$$\tilde{O}(\sum_{i=1}^{m}(\RB+\log (C_{z_i}^{-1})) d (\log(C_{z_i}^{-1})+\tau+d\RB+\Sig)).$$
By rearranging terms, we obtain
$$=\tilde{O}(d^2\RB(\tau+d\RB+\Sig)+d(\tau+d\RB+\Sig)\sum\log (C_{z_i}^{-1})+d(\sum\log (C_{z_i}^{-1}))^2)$$
which equals $\tilde{O}(d(d\RB+\Sig)(\tau+d\RB+\Sig))$  
with Lemma~\ref{lem:sum_of_C_alphas}.
\end{proof}\vspace{0.25cm}

\textbf{Cost of the quadratic sequence.}
Let us fix some root $\xi$ of $f$. Its quadratic sequence consists
of at most $1+\log L$ steps, because $N$ is squared in every step (except
for at most one failing step) and the sequence stops as soon
as the interval is smaller than $2^{-L}$.
Since we ignore logarithmic factors, it is enough to bound the costs of
one QIR step in the sequence.
Clearly, since the interval is not smaller than $2^{-L}$ in such a step,
we have that $\log(b-a)^{-1}\leq L$. Therefore, the required precision
is bounded by $O(L+\tau+d\RB+\Sig)$. It follows that
an \textsc{Aqir} step performs up to $\tilde{O}(d(L+\tau+d\RB+\Sig))$ bit operations.\\

\begin{lemma}\label{lem:quadratic_cost}
The quadratic sequences for one real root is computed within a bit complexity of
$$\tilde{O}(d(L+\tau+d\RB+\Sig)).\\$$
\end{lemma}

\textbf{Total cost.} We have everything together to prove the first main result\\
\begin{theorem}\label{thm:main}
Algorithm~\ref{alg:main} performs root refinement within
$$\tilde{O}(d(d\RB_f+\Sig)^2 + dL)$$
bit operations for a single real root\footnote{In its initial formulation, Algorithm~\ref{alg:main} assumes that isolating intervals for \emph{all} real roots are given. If only one isolating interval $I_k$ for a root $z_k$ is given, we have to normalize $I_k$ first and, then, compute the signs of $f$ at the endpoints of $I$.}  of $f$,
and within
$$\tilde{O}(d(d\RB_f+\Sig)^2 + d^2L)$$
for all real roots.
The coefficients of $f$ need to be approximated to
$\tilde{O}(L+d\RB_f+\Sig)$ bits after the binary point.\\
\end{theorem}

\begin{proof}
We first restrict to the case where $1\le |a_d|<2$. The so far achieved complexity bounds are formulated in terms of an arbitrary (but given) upper bound $\Gamma\in\N$ on $\Gamma_f$. In~\cite[Section 6.1]{sag-complexity}, it is shown how to compute a $\Gamma$ with $\Gamma_f\le\Gamma<\Gamma_f+4\log d$ using $\Otilde((d\Gamma_f)^2)$ bit operations and approximations of $f$ to $\Otilde(d\Gamma_f)$ bits after the binary point. Furthermore, the latter construction also shows that $\tau=\left\lceil \log(\max_i|a_i|)\right\rceil=O(d\Gamma)$ if $1\le |a_d|<2$. 
By Lemma~\ref{lem:normalization_cost}, the normalization for all isolating intervals requires
$\tilde{O}(d(d\RB+\Sig)(\tau+d\RB+\Sig))$
bit operations. The linear subsequences of the \textsc{Aqir} sequence
are computed in the same time by Lemma~\ref{lem:linear_cost}.
The quadratic subsequences are computed with 
$\tilde{O}(d^2L+d^2\tau+d^3\RB+d^2\Sig)$ bit operations 
by Lemma~\ref{lem:quadratic_cost}; the latter three terms are
all dominated by $\tilde{O}(d(d\RB+\Sig)(\tau+d\RB+\Sig))$. 
Hence, with $\Gamma=O(\Gamma_f+\log d)$ as above and $\tau=\Otilde(d\Gamma_f)$, the claimed bound on the bit complexity to refine all roots follows. The maximal number of required bits follows from 
Lemma~\ref{lem:apxqir} because the maximal required
precision in any \textsc{Aqir} step is bounded by $O(L+\tau+d\RB+\Sig)=\tilde{O}(L+d\RB_f+\Sig)$. The bound on refining a single
root follows easily when considering the cost of the quadratic sequence for this root only.

For the more general case, where $1\le|a_d|< 2$ is not necessarily given, we first shift the coefficients by $s=\left\lfloor \log |a_d|\right\rfloor$ bits such that we can apply the above result to the shifted polynomial. Since this coefficient shift does not change the roots, our bit complexity bound follows immediately. For the required precision, we need $\tilde{O}(L+d\RB_f+\Sig)-s$ since we need an approximation of the shifted polynomial to $\tilde{O}(L+d\RB_f+\Sig)$ bits after the binary point.
\end{proof}\vspace{0.25cm}

\section{Experimental Results}
\label{sec:experiments}

We compare the asymptotic bounds of \textsc{Eqir} and \textsc{Aqir}
and their practical behavior for increasing input sizes in the case of integer polynomials.
We have implemented
both algorithms exactly as described in this paper (without the techniques presented in
the forthcoming Section~\ref{sec:tricks}), in the context of the 
\textsc{Cgal}\footnote{Computational Geometry Algorithms Library, \url{www.cgal.org}} library, written in C++.
We used \textsc{gmp}, version 5.0.4, for integer and rational arithmetic.
We generated integer polynomials of various types (described below) using the Maple routine \textit{randpoly},
isolated their real roots using Descartes method, and measured the time to refine them to a predefined
refinement precision on a laptop with dual Pentium core clocked at 2.4 GHz with 3MB cache size each, 
and a total RAM of 4 GB, running Debian squeeze.
Both the source code and the benchmark instances can be sent on request.

In the first run, we chose polynomials with $20$-bit-coefficients chosen uniformly at random 
and a degree between $50$ and $1600$. 
The refinement quality was set to 10000 bits after the binary point.
Table~\ref{tbl:dense} (top) lists the results. 
We also generated two bivariate dense polynomials, each with randomly chosen
$10$-bit coefficients and total degrees between $5$ and $40$, and computed the resultant of them.
The results are listed in Table~\ref{tbl:dense} (bottom).

\begin{table}
\begin{center}
\begin{tabular}{|c|c|c|c|c|c|c|}
\hline
  & \multicolumn{2}{c|}{\textsc{Eqir}} & \multicolumn{3}{c|}{\textsc{Aqir}} & \\
d & $\frac{\text{\# bis.}}{\text{\# roots}}$ & $\frac{\text{time}}{\text{\# roots}}$ & $\frac{\text{\# bis-norm}}{\text{\# roots}}$ & $\frac{\text{\# bis-refine}}{\text{\# roots}}$ & $\frac{\text{time}}{\text{\# roots}}$ & $\frac{t_{\text{exact}}}{t_{\text{approx}}}$ \\
\hline\hline
  50 & 2.5   &   0.438 & 1.5 & 3     &   7.12    &     0.0615 \\
 100 & 1.5 &     1.63  & 1   & 2.5   &  14.2     &     0.115  \\
 200 & 3.5 &     6.40  & 3   & 3     &  30.7     &     0.209  \\
 400 & 3.6 &    24.2   & 2.4 & 2.3   &  60.0     &     0.403  \\
 800 & 3   &    97.6   & 2   & 1.3   & 124       &     0.790  \\
1600 & 4.3 &   392     & 2.3 & 2.3   & 249       &     1.58   \\
\hline
\end{tabular}

\begin{tabular}{|c|c|c|c|c|c|c|c|}
\hline
&  & \multicolumn{2}{c|}{\textsc{Eqir}} & \multicolumn{3}{c|}{\textsc{Aqir}} & \\
$(d_1,d_2)$ & $\tau$ & $\frac{\text{\# bis.}}{\text{\# roots}}$ & $\frac{\text{time}}{\text{\# roots}}$ & $\frac{\text{\# bis-norm}}{\text{\# roots}}$ & $\frac{\text{\# bis-refine}}{\text{\# roots}}$ & $\frac{\text{time}}{\text{\# roots}}$ & $\frac{t_{\text{exact}}}{t_{\text{approx}}}$ \\
\hline\hline
(10,5)  &  161 & 1     &    0.445  & 1     & 0.5   &   7.18     &     0.0620 \\
(10,10) &  226 & 3.3   &    1.67   & 1.8   & 2.5   &  14.6      &     0.114  \\
(20,10) &  353 & 2.2   &    6.38   & 1.8   & 2.5   &  30.1      &     0.212  \\
(20,20) &  487 & 1.8   &   25.2    & 2.2   & 1.7   &  60.1      &     0.414  \\
(40,20) &  755 & 2.9   &  104      & 1.8   & 2.1   & 127        &     0.813  \\
(40,40) & 1042 & 3.6   &  426      & 1.8   & 1.9   & 274        &     1.556  \\
\hline
\end{tabular}
\end{center}
\vspace{0.3cm}
\caption{Experimental results for polynomials with random $20$-bit coefficients 
(first table) and for resultants of bivariate polynomials with random $10$-bit coefficients
(second table). For the latter, the degree is $d_1\cdot d_2$, and the maximal coefficient
bitsize is displayed in the second column. In all cases, the final
precision $L$ is set to $10000$. 
For each degree, we generated $5$ instances and measured the time
of root refinement for \textsc{Eqir} and \textsc{Aqir}. The displayed numbers refer
to the instance whose quotient of running times (last column) 
is the median among the $5$ instances.
The other columns display (from left to right) the number of bisections
the \textsc{Eqir} method performs internally per root, 
the refinement time of \textsc{Eqir} per root, the number of bisections in the normalization
of \textsc{Aqir} per root, the number of bisections \textsc{Aqir} performs per root
after normalizing,
the refinement time of \textsc{Aqir} per root, and the ratio of the total running times
of \textsc{Eqir} and \textsc{Aqir}.}
\label{tbl:dense}
\end{table}

First of all, the quotient between the running times for \textsc{Eqir} and
\textsc{Aqir} is proportional to $d$ which matches the asymptotic bound proved in this paper.
Moreover, both in the exact and approximate version, only a small number of bisections are performed during
the refinement. That means that quadratic convergence takes place almost immediately. The normalization phase
(which only exists in the approximate version) also performs just a small number of bisections.
This implies that the normalization phase and the linear sequence have a minor impact on the practical
running time of the algorithm, and that the cost is dominated by the
quadratic sequence. Recall that a single root can be refined in 
$\Otilde(d^3\tau^2+dL)$ with \textsc{Aqir}, where the first term is caused by
the normalization and the linear sequence, the second by the quadratic sequence.
Indeed, the running time per root increases linearly in $d$ for the approximate variant, as suggested
by the second term of the complexity bound.
For \textsc{Eqir}, the complexity is $\tilde{O}(d^4\tau^2+d^2L)$ for a single root,
with the second term accounting for the quadratic sequence.
We can observe that
the running time per root grows quadratically with $d$. Note that also in the second table,
the running times of both QIR versions are only moderately worse despite the coefficient growth
of the input instances.

We investigate the dependance on the refinement precision $L$ by
fixing a degree of 100 and a coefficient bitsize of $20$, 
and let the final precision $L$ grow from $2000$ to
$128000$ (Table~\ref{tbl:precision}). As we can observe, the quotient of the running times 
of both refinement variants stabilizes for high values of $L$.
However, the growth factor of the running time is not linear in $L$; we observe that the running time roughly
increases by a factor of about $2.6$ when $L$ doubles, 
which corresponds to a growth of roughly $L^{1.4}$.
To explain this super-linear behavior, 
we remark that our analysis ignored logarithmic factors in $L$; 
at least one such factor is included from fast integer arithmetic. Also, \textsc{gmp} does only switch
to asymptotically fast arithmetic for very large integers and uses asymptotically inferior methods
for smaller instances. 

\begin{table}
\begin{center}

\begin{tabular}{|c|c|c|c|}
\hline
$L$ & $\frac{t_{\text{exact}}}{\text{\# roots}}$ & $\frac{t_{\text{approx}}}{\text{\# roots}}$ & $\frac{t_{\text{exact}}}{t_{\text{approx}}}$ \\
\hline\hline
  2000 &  0.0817 &   1.68 & 0.0486 \\
  4000 &  0.220  &   2.89 & 0.0760 \\
  8000 &  0.605  &   6.06 & 0.100  \\
 16000 &  1.60   &  13.9  & 0.115  \\
 32000 &  4.36   &  35.6  & 0.122  \\
 64000 & 11.6    &  94.6  & 0.122  \\
128000 & 28.9    & 242    & 0.120  \\
\hline
\end{tabular}
\end{center}
\vspace{0.3cm}
\caption{Experimental results for polynomials with random $20$-bit coefficients of
degree $100$. Again, the table lists the median over $5$ independent instances. }
\label{tbl:precision}
\end{table}

Finally, we investigate the case of growing coefficient sizes. For that, we fix a degree of $100$ and
a final precision of $10000$ bits and vary the coefficient size. We see in Table~\ref{tbl:coefficients}
that the running time grows very moderately for increasing bitsizes. Also, \textsc{Aqir}
handles large coefficients worse than \textsc{Eqir} (we also have tested a polynomial with $128000$-bit
coefficients where the ratio drops to about $0.01$). 
Recall that our implemented version of \textsc{Aqir} uses absolute precision arithmetic and therefore does not round
the coefficients during the computation.
Consequently, it suffers from high coefficient sizes in every step where it uses
interval arithmetic. An improved version of \textsc{Aqir} using relative precision would remove this drawback.

\begin{table}
\begin{center}
\begin{tabular}{|c|c|c|c|}
\hline
$\tau$ & $\frac{t_{\text{exact}}}{\text{\# roots}}$ & $\frac{t_{\text{approx}}}{\text{\# roots}}$ & $\frac{t_{\text{exact}}}{t_{\text{approx}}}$ \\
\hline\hline
  40 &  1.63  &  14.2  & 0.115  \\
  80 &  1.64  &  14.1  & 0.116  \\
 160 &  1.66  &  14.5  & 0.114  \\
 320 &  1.68  &  14.9  & 0.112  \\
 640 &  1.75  &  15.5  & 0.111  \\
1280 &  1.82  &  16.7  & 0.109  \\
2560 &  1.96  &  19.4  & 0.101  \\
5120 &  2.11  &  24.7  & 0.085  \\
\hline
\end{tabular}
\end{center}
\vspace{0.3cm}
\caption{Experimental results for polynomials of
degree $100$, and a refinement precision of $10000$ bits,
with coefficients chosen uniformly at random.
Again, the table lists the median over $5$ independent instances. }
\label{tbl:coefficients}
\end{table}

\vspace{0.1cm}

To summarize our experiments, the cost of the quadratic sequence dominates the refinement process,
and the cost of this sequence is proportional to $dL^\alpha$ for \textsc{Aqir} and proportional to $d^2L^\alpha$
for \textsc{Eqir} in practice, with $\alpha\approx 1.4$. It shows that the approximate version
is not just a theoretical trick to reduce the complexity, but has a practical impact. 

On the possible disagreement that
\textsc{Aqir} is faster than \textsc{Eqir} only for quite large values of $d$ and $L$,
we reply that our version of \textsc{Aqir} is rather designed for a simple complexity analysis than
for a fast implementation. 
Some optimizations include to use relative instead of absolute precision, to leave out the additional
subdivision points at $m^*\pm\frac{7}{8}\omega$ (which are formally needed for quadratic convergence,
but should be insignificant in practice), and to choose the internal working precision more adaptively
(instead of always setting $\rho\gets 2$ before each while-loop).
We believe that such improvements lead to an implementation
which shows its strength for much smaller instances.

Finally, we remark that the recently introduced \textsc{cgal}-package
on algebraic computations~\cite{bhk-generic} represents algebraic numbers by their
isolating intervals and uses QIR to refine them. The implemented version therein can be considered
as a ``light version'' of the techniques presented in this paper, using relative approximations
to speed up polynomial evaluations, but falling back to exact methods in the case of failure.
The results of our experimental evaluations motivate an integration of a fully approximate variant 
(that is, \textsc{Aqir} with the described optimizations)
into \textsc{Cgal}.

\section{Asymptotic improvements}
\label{sec:tricks}
We further improve our bound from Theorem~\ref{thm:main} in two ways: first, in Section~\ref{ssec:multipoint}, we adapt the technique of \emph{fast multipoint evaluation}
to lower the second term in the bound for \emph{all} real roots from $\tilde{O}(d^2L)$
to $\tilde{O}(dL)$ (so that refining all roots has the same complexity as refining a single root).
Second, we restrict our attention to integer polynomials; the improved bound from Section~\ref{ssec:multipoint}
yields $\tilde{O}(d^3\tau^2+dL)$ for a polynomial of degree $d$ and bitsize $\tau$. Using a recent algorithm for computing isolating
intervals and further refining them to a fixed precision, we improve the latter bound to $\tilde{O}(d^3\tau+dL)$ in Section~\ref{ssec:integers}.
We remark that both optimizations require adaptions of our \textsc{Aqir} algorithm which are not recommended for a practical
implementation (at least not for polynomials of degree and bitsize as they are considered these days).

\subsection{Fast multipoint evaluation}
\label{ssec:multipoint}

It is well known that, roughly speaking, evaluating a univariate polynomial of degree $d$ in $O(d)$ positions simultanuously has the
same arithmetic complexity as evaluating it at a single position, up to logarithmic factors~\cite[Corollary 10.8]{gg-mca-99}. 
These techniques are called \emph{fast multipoint evaluation}; it suggests itself to apply them on our 
\textsc{Aqir} algorithm since polynomial evaluation is the dominant operation. However, since all our evaluations are only approximate
with a fixed working precision, we need an approximate variant of fast multipoint evaluation. 
We use a recent result by Kobel and Sagraloff:\\

\begin{theorem}\cite[Thm. 10]{ks-fast}
\label{thm:ame}
Let $F\in\mathbb{C}[x]$ be a polynomial of degree $d$ with $\|F\|_1\leq 2^\tau$, with $\tau\geq 1$, and let $x_1,\ldots,x_d\in\mathbb{C}$
be complex points with absolute values bounded by $2^\Gamma$, where $\Gamma\geq 1$. Then, approximate multipoint evaluation up to a precision
of $2^{-L}$ for some integer $L\geq 0$, that is, computing $\tilde{y}_j$ such that $|\tilde{y}_j-F(x_j)|\leq 2^{-L}$ for all $j$, is possible
with
$$\tilde{O}(d(L+\tau+d\Gamma))$$
bit operations. The precision demand on $F$ and the points $x_j$ is bounded by $L+O(\tau+d\Gamma+d\log d)$ bits after the binary point.\\\\
\end{theorem}
Note that, with the notations of the theorem, $[\tilde{y}_j-2^L, \tilde{y}_j+2^L]$ is guaranteed to contain $F(x_j)$; 
therefore, the theorem gives an alternative to interval arithmetic with bounded precision.
Specifically, we can replace the usage of interval arithmetic in line~\ref{algstep:ia1} of Algorithm~\ref{alg:abisect}
and in lines~\ref{algstep:ia2} and~\ref{algstep:ia3} of Algorithm~\ref{alg:aqir} by the multipoint evaluation algorithm in~\cite{ks-fast} 
(for now, just applied at a single point). The precision quality is adaptively increased during the execution of the while loop,
and we can prove the same asymptotic bounds (up to an additional term $O(d\log d)$) on the maximal precision 
as in Lemmas~\ref{lem:apxbisection} and~\ref{lem:apxqir}.

Of course, we want to exploit that Theorem~\ref{thm:ame} bounds the cost of evaluating a polynomial at multiple points.
For that goal, we adapt our root refinement algorithm as follows: think of multipoint evaluation as a virtual machine with $d$ input slots
and $d$ output slots which returns $\tilde{y}_1,\ldots,\tilde{y}_d$ for input $x_1,\ldots,x_d$ as described in Theorem~\ref{thm:ame}.
The idea is to perform the refinement of all real roots simultaneously and to use that machine whenever a polynomial has to be evaluated.
To be a bit more precise, reconsider Algorithm~\ref{alg:main}. We leave the normalization subprocedure unchanged (we could use
multipoint evaluation here as well, but it would not change the complexity). Instead of the for-loop, we initialize an integer $P$ to $1$, find all isolating intervals of length at least $2^{-P}$ and call a modified version of \textsc{Aqir} for them that we describe below; 
if all intervals are smaller, we double $P$ and repeat. That means that intervals which are comparably very small 
are not further refined until the other intervals are roughly of the same size.

The modification of \textsc{Aqir} are as follows: we apply Algorithm~\ref{alg:aqir} to all isolating intervals and divide them into two groups:
those for which $N=2$ (that is, an approximate bisection is performed) and those for which $N> 2$. For the first group, we execute
the while loop of Algorithm~\ref{alg:abisect} simultanuously for all intervals; it makes sense to think about this as a parallel process
with execution branches~-- we can easily simulate parallelism by a sequential algorithm that cycles through the different branches.
Every branch fills one input slot of the virtual machine and then waits for the other branches to fill their slots, 
(or send a signal that they have left the loop already).
Once all slots are filled, the machine starts the evaluation and all branches continue their execution until the next loop iteration
requires an evaluation with increased precision. This process continues until all branches have left the loop. 
For the group of intervals that perform an \textsc{Aqir} step with $N>2$, the same strategy is used.

Regarding the complexity of the described method; note that all computations except for the calls of the virtual machine are negligible.\footnote{We remark that, for each \textsc{Aqir} step, we also have to compute an approximation of the fraction of the values $f(a)$ and $f(b)-f(a)$, provided that sufficiently good approximations of $f(a)$ and $f(b)$ are already computed. The cost for the computation of one fraction is then bounded by $\tilde{O}(n\Gamma+\tau+\rho)$, where $\rho$ denotes the required output precision. Hence, when processing up to $n$ intervals in parallel, the total cost is bounded by $\tilde{O}(n(n\Gamma+\tau+\rho))$ bit operations which matches the complexity for one call of the virtual machine with output precision $\rho$.} 
Moreover, for a fixed value of $P$ (as defined above), every interval of length at least $2^{-P}$ is refined by at least one half
per iteration (in an amortized sense).
It follows that there are at most $O(P)$ iterations of the modified \textsc{Aqir} procedure, and afterwards, all intervals
are of size at most $2^{-P}$. On the other hand, if all intervals have entered the quadratic sequence, the virtual machine spends at most $O(\log P)$ 
iterations before doubling $P$ because there is at most one failing QIR call per isolating interval.

We analyze the complexity similar to Section~\ref{sec:aqir_sequence}: Set $C:=\max_\xi C_\xi$, where the maximum is taken over all real roots $\xi$ of $f$ and $C_\xi$ is defined 
as in Definition~\ref{def:C_bound}. Let $P_0$ be the smallest power of two that is larger than $2^{-C}$. 
We bound the complexity to refine all intervals to size $2^{-P_0}$ or less: As we said above, we need $O(P_0)$ calls of the multipoint version
of \textsc{Aqir} for that. Each call, in turn, requires at most
$$\tilde{O}(d(P_0+\tau+d\RB+\Sig))$$
bit operations (compare Lemma~\ref{lem:apxqir} and Theorem~\ref{thm:ame}). Since $P_0\leq 2C\leq 2\sum_\xi C_\xi=O(d(\RB+\log d)+\Sig))$ (Lemma~\ref{lem:sum_of_C_alphas}),
the cost of refining all intervals to size less than $2^{-P_0}$ is bounded by
$$\tilde{O}(d(d\RB+\Sig)(\tau+d\RB+\Sig)$$
with the same argumentation as in Lemma~\ref{lem:linear_cost}.

The benefit of multipoint evaluation takes effect in the second part of the complexity analysis:
suppose that all intervals have entered the quadratic sequence, then, as mentioned above, there
are at most $O(\log P)$ calls per $P$, and there are only $\log L$ different $P$-values reached during the refinement. It follows
that (up to logarithmic factors) the cost is determined by a single execution of the multipoint version of \textsc{Aqir} which is
$$\tilde{O}(d(L+\tau+d\RB+\Sig)).$$
Notice that this matches the previous cost of the quadratic sequence for a single root. Putting everything together, we can prove in analogy 
to Theorem~\ref{thm:main} that the multipoint evaluation variant of \textsc{Aqir} needs
$$\tilde{O}(d(d\RB+\Sig)^2+dL)$$
bit operations to refine \emph{all} isolating intervals to a width of at most $2^{-L}$.

In summary, we obtain the following result:\\

\begin{theorem}\label{thm:themainresult}
When using fast approximate multipoint evaluation, Algorithm~\ref{alg:main} performs root refinement within
$$\tilde{O}(d(d\RB_f+\Sig)^2 + dL)$$
bit operations for all real roots of $f$.
The coefficients of $f$ need to be approximated to
$\tilde{O}(L+d\RB_f+\Sig)$ bits after the binary point.\\
\end{theorem}

\subsection{Integer polynomials}
\label{ssec:integers}

We now concentrate on the case where the polynomial $f$ has integer coefficients of absolute value bounded by $2^\tau$. 
Directly applying Theorem~\ref{thm:themainresult} to $f$, with $\Gamma_f=O(\tau)$ and $\Sig=\Otilde(d\tau)$~\cite[\S 7.2]{sag-complexity}, yields $\tilde{O}(d^3\tau^2+dL)$ for the bit 
complexity of approximating all real roots to an error of $2^{-L}$ or less. 
The quadratic appearance of $\tau$ in the first term is due to the normalization phase and the 
linear sequences; according to Lemma~\ref{lem:quadratic_cost} and Section~\ref{ssec:multipoint}, the 
quadratic sequence for all root amounts for 
$\tilde{O}(d(L+\tau+d\RB+\Sig))=\tilde{O}(d^2\tau+dL)$ bit operations. The 
higher computational cost with respect to $\tau$ for the first two subroutines 
is caused by initial bisection steps before quadratic convergence can be guaranteed. 
The following considerations which already appeared in an extended (unpublished) version of~\cite{NewDsc} show that the normalization phase as well as the linear sequences 
can be replaced by a smarter approach for integer polynomials. 
As a result, the first term in the above complexity bound improves by a factor of $\tau$.

Recent work~\cite{NewDsc} introduces a novel exact subdivision algorithm (denoted \ND) to 
isolate the real roots of a polynomial with integer coefficients. \ND~combines 
Descartes' Rule of Signs, Newton iteration and a QIR-like subdivision strategy, thus achieving quadratic 
convergence for most iterations. In order to keep our 
presentation self-contained, we briefly review the algorithm and refer to the 
full paper for details.\vspace{0.25cm}
\hrule\vspace{0.25cm} 
\ND~recursively subdivides an initial interval $\mathcal{I}_0$ known to contain 
all real roots of $f$ (e.g.~$\mathcal{I}_0=(-2^{\tau+1},2^{\tau+1})$). In each 
iteration, we proceed an interval $I=(a,b)\subset I_0$ and an integer $N_I$, 
where we initially set $N_{\mathcal{I}_0}:=4$. Based on Descartes' 
Rule of Signs, we compute an upper bound\footnote{$v_I$ is the number of sign variations in the coefficient sequence of the polynomial $f_I(x):=(x+1)^d\cdot f((ax+b)/(x+1))$.} $v_I=\operatorname{var}(f,I)$ for the number $m_I$ of roots 
within $I$ which has the same parity as $m_I$. 
If $v_I=0$, 
we discard $I$. If $v_I=1$, we store $I$ as an isolating interval. For 
$v_I>1$, we consider a point $t\in I$ (e.g.~$t=a$ or $t=b$; cf.~\cite{NewDsc} 
for details) and 
compute the Newton approximation $\lambda=t-v_I\cdot \frac{f(t)}{f'(t)}$ 
according to the ``virtual multiplicity'' $v_I$ of $I$. In the case where 
$\lambda\in 
I$, we consider a subinterval $I'=(a',b')$ of width $w(I')=w(I)/N_I$ that 
contains $\lambda$ and compute $v_{I'}:=\operatorname{var}(f,I')$. If 
$v_I=v_{I'}$, we proceed with $I'$ (i.e.~$I\backslash I'$ is discarded) and 
set 
$N_{I'}:=N_I^2$. Otherwise, $I$ is subdivided into two equally sized intervals 
$I_l=(a,\operatorname{mid}(I))$ and $I_r=(\operatorname{mid}(I),b)$, and we 
set $N_{I_l}:=N_{I_r}:=\max(4,\sqrt{N_I})$. $\ND$ proceeds in this way until either all intervals are discarded or stored as isolating.\vspace{0.25cm}
\hrule\vspace{0.25cm}

The complexity analysis from~\cite{NewDsc} shows that $\ND$ isolates all real 
roots of $f$ using no more than 
$\tilde{O}(d^3\tau)$ bit operations. Each of the isolating
intervals $I$ contains exactly one root $\xi$, and it holds that $v_I=1$. 
Notice that 
we can also use $\ND$ for further refining such an 
isolating interval to a width of $2^{-L}$ or less, that is, $I$ is processed 
in the same manner as in 
the isolation routine, but we do not stop until $w(I)<2^{-L}$. The proof 
of~\cite[Theorem~6]{NewDsc} shows that only $O(\log(d\tau)+\log L)$ 
iterations are necessary in order to do so. The cost for each refinement step 
is bounded by $\tilde{O}(d^2(L+\tau))$ since, for computing $v_I$ and the Newton approximation $\lambda$, we have to perform $O(d)$ arithmetic 
operations with $O(d(L+\tau))$ bit numbers; cf.~\cite{NewDsc} for details. Hence, the cost in order to obtain an 
approximation of $\xi$ to $L$ bits after the binary point is bounded by 
$\tilde{O}(d^2(L+\tau))$, and thus $\tilde{O}(d^3(L+\tau))$ for \emph{all real 
roots} of $f$. 
When $L$ is dominating, the latter bound is by a factor of $d^2$ 
larger than the bound $\tilde{O}(d^3\tau^2+d^2L)$ achieved by the 
$\textsc{Aqir}$-method. This is explained by the following two facts: First, $\ND$ exclusively 
uses exact arithmetic, whereas \textsc{Aqir} uses approximate 
arithmetic. Second, we can use fast approximate multipoint evaluation for the \textsc{Aqir} method, whereas $\ND$ uses polynomial evaluation only at single points.  

We design a hybrid method combining $\ND$ and \textsc{Aqir} 
by altering our root refinement strategy as follows: After having isolated the roots using $\ND$, 
we keep refining with the same method until the interval is so small that quadratic convergence
of \textsc{Aqir} is guaranteed. The accumulated cost for getting these intervals
is bounded by $\tilde{O}(d^3\tau)$ as shown below. Then, we apply our the modified \textsc{Aqir} method from Section~\ref{ssec:multipoint} until the interval is smaller
then $2^{-L}$. The cost for that is determined by the quadratic sequence of \textsc{Aqir} which is
$\tilde{O}(d^3\tau+dL)$ for all real roots in total.
The next theorem gives the detailed analysis of this method. The main challenge is that the threshold
for switching from $\ND$ to \textsc{Aqir} depends on parameters which are not readily known; the algorithm
has to estimate these parameters closely enough to achieve the desired asymptotic bound.\\

\begin{theorem}\label{thm:maininteger}
For a square-free polynomial $f$ of degree $d$ with integer coefficients of modulus less than $2^\tau$, we can compute isolating intervals (for all real root of $f$) of width less than $2^{-L}$ using $\tilde{O}(d^3\tau+dL)$ bit operations.\\
\end{theorem} 

\begin{proof}
Let $z_1,\ldots,z_m$ denote the real roots of $f$ and let
$\xi=z_i$ be one of them, for which $\ND$ returns the isolating interval $I:=I(\xi)$.
We want to refine $I$ further using $\ND$ to a width for which success of \textsc{Aqir}
is ensured. Such a bound is given in Corollary~\ref{cor:all_works};
for simplicity, however, we can use the simpler bound of the conference version~\cite{ks-complexity}
of this paper instead: for a normal interval of width
\begin{align}
w(I)<w_{\xi}:=\frac{|f'(\xi)|}{32ed^{3}2^{\tau}\max\{|\xi|,1\}^{d-1}}
\label{sizewI},
\end{align}
where $e\approx 2.71\ldots$ denotes the Eulerian number, it is guaranteed that 
each \textsc{Aqir} step succeeds. 
In order to check whether the inequality 
(\ref{sizewI}) holds, 
the algorithm needs to estimate $|f'(\xi)|$ and 
$\operatorname{max}\{1,|\xi|\}^{d-1}$. For this purpose, 
it uses $\ND$ to refine $I$ further until $w(I)<1/(2d)$ and 
$v_{I^+}:=\operatorname{var}(f,I^+)=1$, where $I^+$ is defined as the enlarged 
interval 
$$I^+=(a^+,b^+):=\left(a-\frac{w(I)}{2}\cdot(2^{3\lceil\log 
d\rceil+6}-1),b+\frac{w(I)}{2}\cdot(2^{3\lceil\log d\rceil+6}-1)\right).$$
The interval $I^+$ is centered at $I$ and has width $w(I^+)=2^{3\lceil\log d\rceil+6}\cdot w(I)\ge 
64d^3w(I)$. According to the two-circle theorem (e.g. see~\cite[Theorem~1]{NewDsc}), $\operatorname{var}(f,I^+)=1$ holds for sure if $w(I^+)<\sigma(\xi,f)/2$. It 
follows that the endpoints of the so-obtained intervals $I$ and $I^+$ are 
dyadic numbers that can be represented by $O(\tau+\log 
d+\log\sigma(\xi,f)^{-1})$ many bits. Hence, the cost for this refinement is 
bounded by $\tilde{O}(d^2(\tau+\log\sigma(\xi,f)^{-1}))$ bit operations since, 
in each iteration, we perform $O(d)$ arithmetic operations, and the total 
number of iterations is bounded by 
$O(\log(d\tau)+\log\log\sigma(\xi,f)^{-1})$. 
This yields the bound $\tilde{O}(d^3\tau)$ for the total cost for all real roots because $\sum_{i=1}^m \log \sigma(z_i,f)^{-1}=\Sigma_F+O(d\tau)=\tilde{O}(d\tau)$.

Since 
$\operatorname{var}(f,I^+)=1$, the Obreshkoff lens $L_d^+$ for $I^+$ (see~\cite[Figure~2.1]{NewDsc} for the definition and an illustration of the Obrsehkoff lens) contains exactly one root, namely, $\xi\in I$. According to~\cite[Lemma 5]{NewDsc}, the distance from an arbitrary point in $I$ to an arbitrary point outside $L_d^+$ is lower bounded by
$$\frac{1}{4d}\cdot\left(\min\{|a^+ -a|,|b^+ -b|\}-8d^2w(I)\right)>4d^2w(I),$$
and thus $w(I)<\sigma(\xi,f)/(4d^2)$. It follows that each point within $I$ has distance more than $\sigma(\xi,f)/2$ to any root $z_j\neq \xi$.
It is well-known (e.g.~\cite{yap-fpaa-00}) that the disc $\Delta_{\sigma(\xi,f)/d}(\xi)$ of radius $\sigma(\xi,f)/d$ centered at $\xi\in I$ contains no root of the derivative $f'$, hence the disc $\Delta_{2dw(I)}(\operatorname{mid}(I))\subset \Delta_{\sigma(\xi,f)/d}(\xi)$ contains no root of $f'$ either. It follows that 
\begin{align}
|f'(\xi)|/2< |f'(a)|<2|f'(\xi)|\label{firstinequality}
\end{align}
since, for each root $z_j'$ of the derivative $f'$, we have $|a-z_j'|/|\xi-z_j'|\in (1-1/(2d),1+1/(2d))$, and $(1+1/(2d))^{d-1}<\sqrt{e}<2$ and $(1-1/(2d))^{d-1}>1/\sqrt{e}>1/2.$ 
In addition, we have $$(1-1/(2d))\cdot\max\{1,|\xi|\}<\max\{1,|a|\}<(1+1/(2d))\cdot\max\{1,|\xi|\}$$ since $w(I)<1/(2d)$. Hence, it follows that 
\begin{align}
\frac{1}{2}\cdot\max\{1,|a|\}^{d-1}<\max\{1,|\xi|\}^{d-1}<2\max\{1,|a|\}^{d-1}.\label{secondinequality}
\end{align}
Combining (\ref{firstinequality}) and (\ref{secondinequality}) with (\ref{sizewI}), we have $w(I)<w_{\xi}$ if
\begin{align}
w(I)<w_a:=\frac{|f'(a)|}{256ed^{3}2^{\tau}\max\{|a|,1\}^{d-1}}.\label{sizewItest}
\end{align}
In fact, we even have $w(I)<w_{\xi}/2$ in this case, which will turn out useful in the last step.

The algorithm further refines $I$ using $\ND$ until $w(I)<w_a$. 
Using (\ref{firstinequality}) and (\ref{secondinequality}) in the other direction, it follows that $w_a>w_{\xi}/32$;
therefore $I$ is refined to a width of not smaller than $(w_{\xi}/32)^2$. 
These refinements demand for $\tilde{O}(d^2(\tau+d\log\max\{1,|\xi|\}-\log |f'(\xi)|))$ bit operations. 
Using the Mahler bound yields $\log\prod_{i=1}^d \max\{1,|z_i|\}=O(\tau+\log d)$. 
The product of all $f'(z_i)$, $i=1,\ldots,d$, equals $\operatorname{lcf}(f)^{2-d}\operatorname{Disc}(f)$, 
where $\operatorname{lcf}(f)$ denotes the leading coefficient and $\operatorname{Disc}(f)\in\mathbb{Z}$ the discriminant of~$f$. Since $|f'(z)|\le d^2 2^\tau \max\{1,|z|\}^d$ for all $z\in\mathbb{C}$, it follows that 
\begin{equation*}\begin{split}
\prod_{i=1}^m |f'(z_i)|&\ge \prod_{i>m} (d^2 2^\tau \max\{1,|z_i|\}^d)^{-1} \prod_{i=1}^d |f'(z_i)|\ge \prod_{i=1}^d(d^2 2^\tau \max\{1,|z_i|\}^d)^{-1} \operatorname{lcf}(f)^{2-d}\operatorname{Disc}(f)\\
&=2^{-O(d(\log d +\tau))}
\end{split}\end{equation*}
Thus, the total cost for the refinement is bounded by $\tilde{O}(d^3\tau)$.

Notice that, after the latter refinement steps, the width of the interval $I$
satisfies $w(I)<w_{\xi}/2$, but $I$ may not be normal, a property which is required to ensure success of the \textsc{Aqir} steps. The following consideration however shows that the interval $$\tilde{I}=(\tilde{a},\tilde{b}):=(a-w(I)/2,b+w(I)/2)$$ of double width centered at $I$ is normal: Obviously, the first property of Definition~\ref{def:normal} is satisfied for $\tilde{I}$. We have already shown that the distance from an arbitrary point within $I$ to an arbitrary root $z_j\neq \xi$ is more than $\sigma(\xi,f)/2$, and $w(I)<\sigma(\xi,f)/(4d^2)$. Hence, each point $p\in\tilde{I}$ has distance more than $\sigma(z_j,f)/4$ to any root $z_j\neq \xi$. This shows that the second property of Definition~\ref{def:normal} is satisfied. 
Furthermore, both endpoints of $\tilde{I}$ have distance at least $w(\tilde{I})/4$ from $\xi$. A completely analogous computation as in the proof of Lemma~\ref{lem:normalization} then shows that the third property of Definition~\ref{def:normal} holds as well. Thus, considering the interval $\tilde{I}$ as the starting interval for the $\textsc{Aqir}$ method, 
quadratic convergence is achieved for all steps right from the beginning. According to Lemma~\ref{lem:quadratic_cost} and the argumentation in Section~\ref{ssec:multipoint}, the cost for the remaining refinement steps are then bounded by $\tilde{O}(d^3\tau+dL)$
\end{proof}\vspace{0.25cm}

\section{Concluding Remarks}
\label{sec:conclusion}

We have presented a complete solution to the root refinement problem
using validated numerical methods in this paper. Despite the relative
simplicity of the approach, we obtain a bit complexity which is essentially
competitive to best known bounds which have been achieved by much more
sophisticated algorithms. Moreover, we have demonstrated that our approach
is easily implementable and leads to practical improvements even
when implemented in the most naive form.

We have shown that the complexity of approximating roots of a real polynomial only depends on the geometry of the roots and not on the complexity or the type of the coefficients. For instance, we used this fact in~\cite{ks-worst} to derive considerably improved complexity bounds for the topology computation
of algebraic plane curves.

Although the focus of this work was the asymptotic complexity,
the presented algorithm also aims for a practically efficient solution
of the root approximation problem. Indeed, a simplified version of our approach
(for integer coefficients) is included in the recently introduced
\textsc{cgal}-package
on algebraic computations~\cite{bhk-generic}. 
Experimental comparisons in the context of~\cite{bes-bisolve11} have shown
that the approximate version of QIR gives significantly better running times
than its exact counterpart.
These observations underline the practical relevance of our approximate version
and suggest a practical comparison with state-of-the-art solvers as further work.

\end{document}